\documentclass[10pt,twocolumn,fleqn]{article}
\usepackage{geometry}
\usepackage[sc]{mathpazo}
\usepackage[mathscr]{eucal}
\usepackage{enumerate}
\usepackage{booktabs}
\usepackage{tabularx}
\usepackage{hyperref}
\usepackage{multirow}
\usepackage{cite}
\usepackage{mathptmx} 
\usepackage{amsmath,amsfonts,amssymb,amsthm}
\usepackage{caption}
\usepackage{subfigure}
\usepackage{tgpagella}
\usepackage{graphicx}
\usepackage{xcolor}
\usepackage[textwidth=3cm,textsize=tiny,shadow]{todonotes}
\newcommand\scD{{\mathscr D}}

\newcommand\scM{{\mathscr M}}
\newcommand\scN{{\mathscr N}}

\newcommand\scT{{\mathscr T}}

\newcommand\mvector{\boldsymbol}

\newcommand\vv{\mvector{v}}

\newcommand\vx{\mvector{x}}

\newcommand\vA{\mvector{A}}

\newcommand\vX{\mvector{X}}

\newcommand\vvarphi{\mvector{\varphi}}

\newcommand\field{\mathbb}

\newcommand\R{\field{R}}

\newcommand\C{\field{C}}
\newcommand\Z{\field{Z}}

\newcommand\N{\field{N}}
\newcommand\Q{\field{Q}}

\newcommand\rmd{\mathrm{d}}

\newcommand\pder[2]{\dfrac{\partial #1 }{\partial #2}}

\usepackage{mathtools}
\theoremstyle{plain}
\newtheorem{theorem}{Theorem}
\newtheorem{lemma}[theorem]{Lemma}

\newtheorem{corollary}[theorem]{Corollary}
\newtheoremstyle{note}{\topsep}{\topsep}{\slshape}{}{\scshape}{}{ }{}
\theoremstyle{note}

\theoremstyle{remark}
\newtheorem{remark}[theorem]{Remark}
\numberwithin{equation}{section}
\numberwithin{theorem}{section}
\usepackage{enumitem}

\setlist[description]{font=\normalfont}
\usepackage{scalerel,stackengine}
\stackMath
\newcommand\reallywidehat[1]{%
\savestack{\tmpbox}{\stretchto{%
  \scaleto{%
    \scalerel*[\widthof{\ensuremath{#1}}]{\kern.1pt\mathchar"0362\kern.1pt}%
    {\rule{0ex}{\textheight}}
  }{\textheight}%
}{2.4ex}}%
\stackon[-6.9pt]{#1}{\tmpbox}%
}

\title{A new model of variable-length coupled pendulums: from hyperchaos to superintegrability}
\author{
Wojciech Szumi{\'n}ski \\
Institute of Physics, University of Zielona G\'ora, \\
 Licealna 9, PL-65--417,  Zielona G\'ora, Poland
\\ e-mail: w.szuminski@if.uz.zgora.pl
}

\begin{document}
\maketitle
\begin{abstract}
	This paper studies the dynamics and integrability of a variable-length coupled pendulum system. 
	The complexity of the model is presented by joining various numerical methods, such as the Poincar\'e cross-sections, phase-parametric diagrams, and Lyapunov exponents spectra. We show that the presented model is hyperchaotic, which ensures its nonintegrability. We gave analytical proof of this fact analyzing properties of the differential Galois group of variational equations along certain particular solutions of the system. We employ the Kovacic algorithm and its extension to dimension four to analyze the differential Galois group. 
	Amazingly enough, in the absence of the gravitational potential and for certain values of the parameters, the system can exhibit chaotic, integrable, as well as superintegrable dynamics.  To the best of our knowledge, this is the first attempt to
use the method of Lyapunov exponents in the systematic search for the first integrals of the system.  We show how to effectively apply the Lyapunov exponents as an indicator of integrable dynamics.  The explicit forms of integrable and superintegrable systems are given. 

\paragraph{Declaration} The article has been published in \textit{Nonlinear Dynamics}, and the final version is available at this link: \href{https://link.springer.com/article/10.1007/s11071-023-09253-5}{URL} 
\end{abstract}
\section{Introduction  and motivation}
	
Studies of nonlinear dynamics and chaos in pendulum systems are well-established but still are in great scientific activity.  Indeed, one can find numerous papers, books, and video clips concerning their highly nonlinear dynamics~\cite{Baker:05::}. The paradigm models such as  the double pendulum~\cite{Shinbrot:92::,Stachowiak:06::,Stachowiak:15::}, the spring pendulum~\cite{Broucke:73::,Lee:97::,Maciejewski:04::c}, the system of two coupled pendulums~\cite{Huynh2010,Huynh2013,Elmandouh:16::,Szuminski:20::}, the swinging Atwood machine~\cite{Tufillaro:84::,Tufillaro:85::,Tufillaro:90::,Szuminski:22::},  have been broadly studied by many researchers both theoretically and experimentally~\cite{Levien:93::,Kuhn:12::,Pujol:10::}. For instance, the model of coupled pendulums plays a crucial role in the theory of synchronizations~\cite{Dilao:09::,KOLUDA2014977,DUDKOWSKI20181} which have practical applications in laboratory experiments~\cite{PhysRevLett.72.2009, Rosenblum:03::,Ticos:00::,PALUS2007421,Palus:07::}. Moreover,   the system of two coupled pendulums has a direct relation with a  two coupled current-biased Josephson junction~\cite{Aronson:91::,Koyama:96::}, which is meaningful in a field of superconductivity and quantum information~\cite{PhysRevB.43.229,Han:01::}. We also mention papers~\cite{Xie:14::,Kapitaniak:14::,Wojewoda:16::}, where the phenomenon of chimera states in the systems of coupled pendulums was studied.

	In this paper, we want to explore more deeply the dynamics and integrability of a generalized model of coupled pendulums. Namely, it is a combination of a simple coupled pendulum system with the swinging Atwood machine. Thus, it can be treated as a variable-length coupled pendulum as well as the double-swinging Atwood machine with additional Hooke interactions.
	
Such models are of interest due to their potential physical applications in crane models, where understanding the motion and stability is crucial for safe and efficient operation~\cite{JU2006376,MR4459645,Freundlich:20::}, Moreover, the flexibility and maneuverability of the variable length pendulum system make it important in robotics, where dynamic stability is crucial~\cite{PLAUT20133768,YANG2022116727,SHARGHI2022117036}. 
Finally, the combination of a system of pendulums of variable lengths with the swinging Atwood machine may have applications in energy conversion and storage, where the swinging can be used to generate electricity~\cite{MARSZAL2017251,doi:10.1177/14613484221077474,ABOHAMER2023377}.
For a comprehensive review of variable-length pendulums and
their physical realizations please consult the new papers~\cite{Yakubu:21::,Yakubu:22::,Olejnik:23::}.
 
As the proposed model is a Hamiltonian system, its total energy, which is a conserved quantity, determines the global properties of motion. Typically, for relatively low values of energy, we may expect the system's motion to be regular with quasi-periodic and periodic oscillations. However, for sufficiently large values of the energy, the pendulum systems exhibit typically chaotic behaviour\cite{ mp:13::c,Szuminski:14::,Stachowiak:15::,Szuminski:20::,Szuminski:22::}. Complex dynamics in Hamiltonian systems can be effectively visualized with the help of numerical methods such as the Poincar\'e cross sections, phase-parametric (bifurcations) diagrams, Lyapunov's exponents, and power spectra. Each of these methods has its strengths and weaknesses. 
	For instance the Poincar\'e cross-sections
	provide qualitative information about the dynamics by presenting the coexistence of periodic, quasi-periodic, and chaotic motion.  Nevertheless, for technical reasons, it is mostly used for Hamiltonian systems with two degrees of freedom. Although the Lyapunov exponents method is useful for obtaining a quantitative description of chaos and can be effectively applied to a system with many degrees of freedom,  it does not distinguish periodic solutions from quasi-periodic ones. 
Therefore, to gain an exhaustive insight into the dynamics of the considered model,  we combine Lyapunov's exponents spectrums with bifurcation diagrams and the Poincar\'e cross sections.
	
Despite the advantages of numerical methods and techniques, they have one weak point. Namely, each numerical analysis can be performed only for fixed values of parameters describing a system. For pendulum systems, such parameters include the lengths of the pendulum arms, masses of bobs, spring stiffness, etc.
	 For various values of the parameters, the dynamics of the system may be significantly different, and for particular sets, the system may have first integrals, and it can even be integrable.
	This makes the numerical analysis less practical for hunting first integrals. 
	
To find new integrable cases or to prove the nonintegrability of the  considered model, 
	one needs a strong tool. An effective and strong tool is the so-called Morales-Ramis theory~\cite{Morales:99::,Morales:00::}.  It is based on an analysis of the differential Galois group of variational equations obtained by linearization of equations of motion along a  particular solution. The main theorem of this theory states that if a Hamiltonian system is integrable in the sense of Liouville, then the identity component of the differential Galois group of variational equations must be Abelian. The Morales--Ramis theory has  already been successfully applied to various important physical systems~\cite{Yagasaki:18::,Acosta:18::,Acosta:18b::,Huang:18::,Combot:18::,Mnasri:18::,Shibayama:18::,Maciejewski:18::}, also to non-Hamiltonian ones~\cite{Huang:18::,Szuminski:18::,Maciejewski:20e::,Szuminski:20b::}. In this way,  integrable and super-integrable systems have been found~\cite{	Elmandouh:18::,Szuminski:18a::,Szuminski:18b::}. 
	
	In most cases, however, the Morales-Ramis theory has been applied to Hamiltonian systems of two degrees of freedom for which the procedure of analysis of the differential Galois group is known thanks to the  Kovacic algorithm~\cite{Kovacic:86::}. In literature, there is a lack of exhaustive integrability analysis of pendulum systems with many degrees of freedom. This is due to a considerably more complicated analysis of the differential Galois group of high-dimensional variational equations. However, the presented model has a nice property, and an effective integrability analysis via the differential Galois approach and the Kovacic algorithm of dimension four~\cite{Combot:18b::} is possible.
	
	The rest of this paper proceeds as follows. In Sec.~\ref{sec:1} a description of the proposed model and its dynamics is given. We provide a qualitative and quantitive description of chaos and hyperchaos by joining numeral methods, such as Lyapunov's exponent's spectra, phase-parametric diagrams, and the Poincar\'e sections. In Sec.~\ref{sec:2} we perform an effective integrability analysis of the model with the help of the Morales--Ramis theory and the application of the Kovacic algorithms of dimensions two and four. In Sec.~\ref{sec:3} the dynamics and integrability of the coupled pendulum system in the absence of the gravitational potential are treated. Nonintegrability, integrability, and superintegrability for certain sets of parameters of the system are shown. In Sec.~\ref{sec:4}  final comments and conclusions are drawn. Sec.~\ref{sec:5} contains an Appendix in which the  Kimura theorem concerning the solvability of the Gauss hypergeometric differential equations is given. 
	
	\section{\label{sec:1} The system  and its dynamics}
	
	\begin{figure}[t!]
	\begin{center}
		\resizebox{85mm}{!}{\begingroup%
		  \makeatletter%
		\providecommand\color[2][]{%
			\errmessage{(Inkscape) Color is used for the text in Inkscape, but the package 'color.sty' is not loaded}%
			\renewcommand\color[2][]{}%
		}%
		\providecommand\transparent[1]{%
			\errmessage{(Inkscape) Transparency is used (non-zero) for the text in Inkscape, but the package 'transparent.sty' is not loaded}%
			\renewcommand\transparent[1]{}%
		}%
		\providecommand\rotatebox[2]{#2}%
		\newcommand*\fsize{\dimexpr\f@size pt\relax}%
		\newcommand*\lineheight[1]{\fontsize{\fsize}{#1\fsize}\selectfont}%
		\ifx\svgwidth\undefined%
		\setlength{\unitlength}{248.83760708bp}%
		\ifx\svgscale\undefined%
		\relax%
		\else%
		\setlength{\unitlength}{\unitlength * \real{\svgscale}}%
		\fi%
		\else%
		\setlength{\unitlength}{\svgwidth}%
		\fi%
		\global\let\svgwidth\undefined%
		\global\let\svgscale\undefined%
		\makeatother%
  \begin{picture}(1,0.61995447)%
			\put(0.10415659,0.26045207){\color[rgb]{0,0,0}\makebox(0,0)[lt]{\smash{$m_1$}}}%
			\put(0,0){\includegraphics[width=\unitlength,page=1]{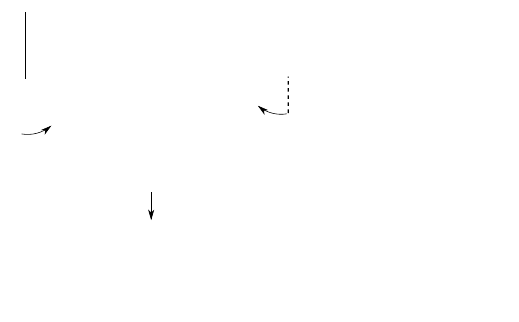}}%
			\put(0.42301579,0.28877563){\color[rgb]{0,0,0}\makebox(0,0)[lt]{\smash{$ m_2$}}}%
			\put(0.26133694,0.33682289){\color[rgb]{0,0,0}\makebox(0,0)[lt]{\smash{$k_2$}}}%
			\put(0.92133216,0.12961847){\color[rgb]{0,0,0}\makebox(0,0)[lt]{\smash{$k_1$}}}%
			\put(0.28112123,0.48963458){\color[rgb]{0,0,0}\makebox(0,0)[lt]{\smash{$a$}}}%
			\put(0.31,0.21710442){\color[rgb]{0,0,0}\makebox(0,0)[lt]{\smash{$g$}}}%
			\put(0.04883007,0.38526667){\color[rgb]{0,0,0}\makebox(0,0)[lt]{\smash{$\vartheta$}}}%
			\put(0.52520297,0.41947521){\color[rgb]{0,0,0}\makebox(0,0)[lt]{\smash{$\varphi$}}}%
			\put(0.12423339,0.37546204){\color[rgb]{0,0,0}\makebox(0,0)[lt]{\smash{$\ell$}}}%
			\put(0.44354039,0.40638972){\color[rgb]{0,0,0}\makebox(0,0)[lt]{\smash{$\ell$}}}%
			\put(0.01597873,0.43860278){\color[rgb]{0,0,0}\makebox(0,0)[lt]{\smash{$x$}}}%
			\put(0.06544636,0.49129843){\color[rgb]{0,0,0}\makebox(0,0)[lt]{\smash{$y$}}}%
			\put(0,0){\includegraphics[width=\unitlength,page=2]{geom.pdf}}%
			\put(0.90874008,0.3627055){\color[rgb]{0,0,0}\makebox(0,0)[lt]{\smash{$X$}}}%
			\put(0.94628674,0.24745895){\color[rgb]{0,0,0}\makebox(0,0)[lt]{\smash{$M$}}}%
			\put(0,0){\includegraphics[width=\unitlength,page=3]{geom.pdf}}%
			\end{picture}%
			\endgroup}
\caption{ (Color online) Geometry of the variable-length coupled pendulums moving in
				the gravitational and Hooke’s potentials. Here $M$ and $m_1,m_2$ are the masses linked by an inextensible string of lengths $l_1=l+a$ (red) and  $l_2=l$ (green). Masses $m_1$ and $m_2$ can swing, whereas $M$ is constrained to move solely in the vertical direction. The Hamiltonian function that describes the model is defined in~\eqref{eq:hamiltonian}. \label{fig:1}}
	\end{center}
	\end{figure}

	In Fig.~\ref{fig:1}, the geometry of the system under consideration is presented. The model consists of three masses $M, m_1, m_2$, two inextensible strings of lengths $l_1$ (red) and  $l_2$ (green), and two springs with Yang's modulus  $k_1, k_2$, respectively.
	The distance between pulleys equals the rest length of spring  $k_2$ and is denoted by $a$. The natural length of the spring with Yang's modulus  $k_1$ is assumed to be zero.
	Masses $m_1$ and $m_2$ are mechanically linked with mass~$M$ and they are allowed to oscillate in a plane. Thus, they form the variable-length two-pendulum system coupled by the spring $k_1$. 
Mass $M$ plays the role of a counterweight and it moves vertically. The pendulums and mass $M$  move under the constant vertical gravitational field and their interactions are facilitated through elastic forces.	

The Lagrange function of the system is as follows
	\begin{equation}
	\label{eq:lagrangian_cartesian}
	\begin{split}
	&L=T-V_g-V_k,\\
	&T=\frac{1}{2}\left(M\dot X^2+m_1\left(\dot x_1^2+\dot y_1^2\right)+m_2\left(\dot x_2^2+\dot y_2^2\right)\right),\\
	&V_g=-g(M X+m_1 x_1+m_2 x_2)\\
	&V_k=\frac{1}{2}k_1\left(l-X\right)^2+\frac{1}{2}k_2\left(\sqrt{\Delta x^2+\Delta y^2}-a\right)^2,
	\end{split}
	\end{equation}
	where $\Delta x=x_2-x_1$ and $\Delta y=y_2-y_1$.
	The motion of the system is restricted by the holonomic constraints, i.e., lengths of the strings are constant
	 \begin{equation}
	\sqrt{x_1^2+y_1^2}+a+X=l_1,\qquad \sqrt{x_2^2+(y_2-a)^2}+X=l_2.
	\end{equation} Therefore,
	\begin{equation}
	\label{eq:XX}
\sqrt{x_1^2+y_1^2}-\sqrt{x_2^2+(y_2-a)^2}=l_1-(l_2+a).
	\end{equation}
To simplify further analysis and reduce the number of parameters,  we assume $l_1=l+a$, and $l_2=l$. Next, we introduce  new coordinates according to the constraints:
	\begin{equation}
	\label{eq:polarki}
	\begin{split}
	x_1&=\ell\cos\vartheta,\quad y_1=\ell\sin\vartheta,\\
	x_2&=\ell\cos\varphi,\quad y_2=a+\ell\sin\varphi,\\
	X&=l-\ell.
	\end{split}
	\end{equation}	
	In these coordinates, the Lagrange function~\eqref{eq:lagrangian_cartesian}, takes the form
	\begin{equation}
		\begin{split}
			\label{eq:lagrangian}
			&L=T-V_g-V_k,\\
			&T=\frac{1}{2}\left((M+m_1+m_2)\dot \ell^2+m_1\ell^2\dot \vartheta^2+m_2\ell^2\dot \varphi^2\right) \\
			&V_g=g\ell(M-m_1\cos\vartheta-m_2\cos\varphi),\\
			&V_k=\frac{1}{2}k_1 \ell^2+\frac{1}{2}k_2(d-a)^2,
		\end{split}
	\end{equation}
	where  $d$ is a length of the second spring  $k_2$, given by
	\begin{equation}
		d:=\sqrt{\left(\ell \cos\varphi-\ell\cos\vartheta\right)^2+\left(a+\ell\sin\varphi-\ell\sin\vartheta\right)^2}.
	\end{equation}
	Performing the Legendre transformation
	\begin{equation}
		\begin{split}
			&p_\ell=\pder{L}{\dot \ell}=(M+m_1+m_2)\dot \ell,\\ & p_\vartheta=\pder{L}{\dot \vartheta}=m_1 \ell^2\dot\vartheta,\\ & p_\varphi=\pder{L}{\dot\varphi}=m_2 \ell^2\dot\varphi,
		\end{split}
	\end{equation}
	we obtain the Hamiltonian function
	\begin{equation}
		\label{eq:hamiltonian}
		\begin{split}
			&H=\dfrac{1}{2}\left(\dfrac{p_\ell^2}{M+m_1+m_2}+\dfrac{p_\vartheta^2}{m_1\ell^2}+\dfrac{p_\varphi^2}{m_2\ell^2}\right)\\
			&	+g\ell(M-m_1\cos\vartheta-m_2\cos\varphi)+\frac{1}{2}k_1 \ell^2+\frac{1}{2}k_2(d-a)^2.\end{split}
	\end{equation}

	The Hamiltonian equations of motion, generated by Hamiltonian~\eqref{eq:hamiltonian}, form a six-dimensional system of the  first-order ordinary  differential equations
	\begin{equation}
	\begin{aligned}
	\label{eq:vh0}
	&\dot \ell=\frac{\partial H}{\partial p_\ell}, && \dot \vartheta=\frac{\partial H}{\partial p_\vartheta},&&\dot \varphi=\frac{\partial H}{\partial p_\varphi},\\
&\dot  p_{\ell}	=-\frac{\partial H}{\partial \ell}, && \dot  p_{\vartheta}	=-\frac{\partial H}{\partial \vartheta}, && \dot  p_{\varphi}	=-\frac{\partial H}{\partial \varphi}.
	\end{aligned}
	\end{equation}
	The explicit forms of the right-hand sides of~\eqref{eq:vh0} are as follows
	\begin{equation}
		\begin{cases}
			\label{eq:vh}
			\dot \ell=\dfrac{p_\ell}{M+m_1+m_2},\\[0.3cm]
			\dot p_\ell=\dfrac{m_2p_\vartheta^2+m_1p_\varphi^2}{m_1m_2\ell^3}-g(M-m_1\cos\vartheta-m_2\cos\varphi)\\[0.3cm]
			-k_1\ell-k_2\left[d-a\right]\left[a(\sin\varphi-\sin\vartheta)+2\ell(1-\cos(\vartheta-\varphi))\right]/d,
			\\[0.2cm] \dot\vartheta=\dfrac{p_\vartheta}{m_1\ell^2},\\[0.3cm] \dot p_{\vartheta}=-m_1 g \ell \sin\vartheta+k_2\ell\left[d-a\right]\left[a \cos\vartheta-\ell\sin(\vartheta-\varphi)\right]/d,\\[0.2cm] \dot\varphi=\dfrac{p_\varphi}{m_2\ell^2},\\[0.3cm]
			\dot p_{\varphi}=-m_2 g \ell \sin\varphi-k_2\ell\left[d-a\right]\left[a \cos\varphi-\ell\sin(\vartheta-\varphi)\right]/d.
		\end{cases}
	\end{equation}
	\subsection{The Lyapunov exponents diagrams}
	\begin{figure}[htp]
		\centering                              
		\includegraphics[width=.95\linewidth]{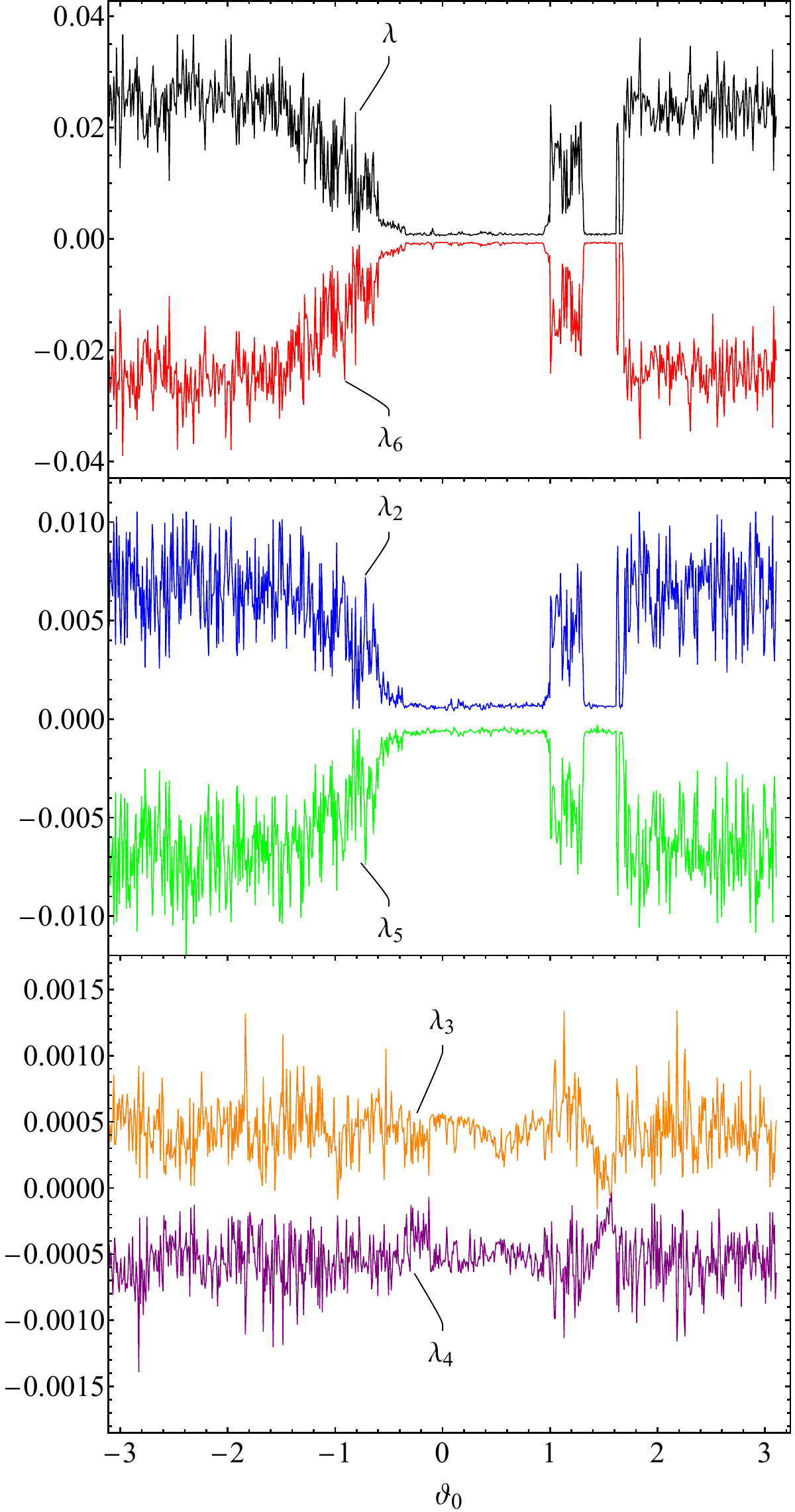}\hspace{-0.2cm}
		\caption{(Color online) The Lyapunov exponents spectrum of system~\eqref{eq:vh} versus the initial swinging angle $\vartheta_0\in (-\pi,\pi)$ with the initial condition $(\ell_0=1,\, \varphi_0=\pi/2,\,	p_{\ell 0}=0.002, \, p_{\vartheta 0}=0.001, \, p_{\varphi 0}=0.001)$. The constant parameters where chosen as: $M=2,\, m_1=2,\, m_2=1,\, a=5,\, g=1,\, k_1=0.1,\, k_2=0.25$. Here $\{\lambda,\lambda_2,\lambda_3,\lambda_4,\lambda_5,\lambda_6\}$ denotes the full spectrum, where $\lambda\equiv \lambda_1$ is the largest Lyapunov exponent. Intervals with two positive Lyapunov exponents are responsible for the hyperchaotic motion of the system, while regions with $\lambda\approx 0$ correspond to regular (non-chaotic) behavior. \label{fig:spectrum}}
	\end{figure}
	\begin{figure*}[t]
	\centering                              
	\includegraphics[width=.35\linewidth]{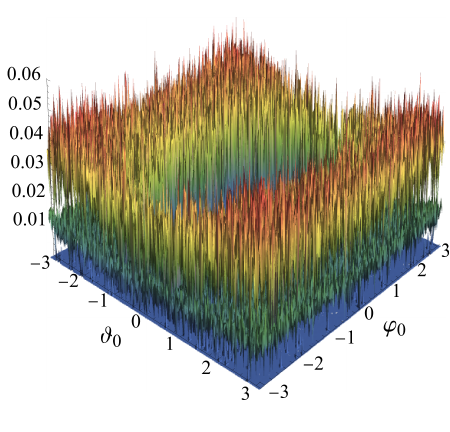}\hspace{-0.2cm} \
	\includegraphics[width=.3\linewidth]{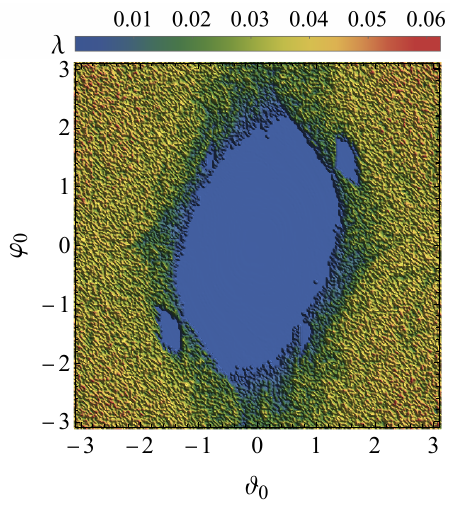}\ 
	\includegraphics[width=.3\linewidth]{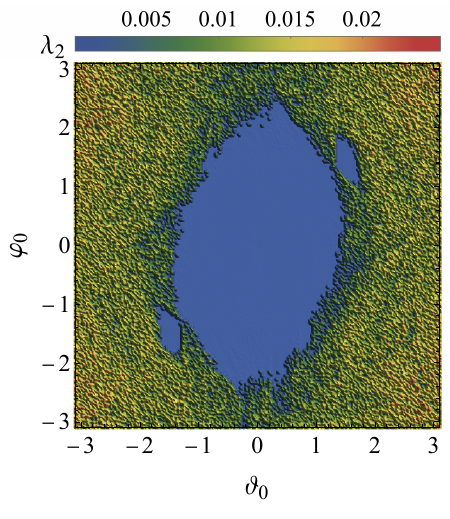}\hspace{-0.2cm}
	\caption {(Color online) Three-dimensional Lyapunov's exponents diagram of system~\eqref{eq:vh} depicted in $(\vartheta_0,\varphi_0,\lambda)$-space  and  the projections of $\lambda$ and $\lambda_2$  to $(\vartheta_0,\varphi_0)$-plane.	 The colorful diagram was obtained by numerically computing  Lyapunov's exponents on a grid of $500\times 500$ values of $(\vartheta_0,\varphi_0)$ taken over the range $(-\pi,\pi)$ with the initial condition $(\ell_0=1,\,	p_{\ell 0}=0.002,\, p_{\vartheta 0}=0.001,\, p_{\varphi 0}=0.001)$.  Exemplary constant parameters where chosen as: $M=2,\, m_1=2,\, m_2=1,\, a=5,\, g=1,\, k_1=0.1,\, k_2=0.25$.	
		 The central part of the diagram corresponds to regular (non-chaotic) dynamics, while for larger values of initial swing angles, the hyperchaotic motion takes place.	\label{fig:lyap3da}}
\end{figure*}
In this section, we present the complexity of the system and we study its hyperchaotic nature through the analysis of Lyapunov's exponents. 
The method of Lyapunov exponents is an essential tool for quantifying chaos in dynamical systems. It measures the exponential divergence of two close trajectories (orbits) in a phase space. According to the chaos theory, chaos appears when one Lyapunov's exponent is positive, while hyperchaos is characterized by the presence of at least two positive Lyapunov exponents~\cite{Szuminski:18::,Szuminski:20b::}.	For the computation of Lyapunov exponent spectra, we employ the standard algorithm introduced by Benettin et al~\cite{Benettin:80::,Wolf:85::}. It is based on successive integrations of variational equations with applications of the Gram-Schmidt orthonormalization procedure.  

In this paper, we utilize the standard algorithm implemented in Mathematica by Sandri~\cite{Sandri:96::}. However, for
faster and more accurate results, we employ the NDSolve solver instead of Euler's method. We adopt a sufficient amount of~$k$ steps so that the convergence of the Lyapunov exponents is ensured. 
The working precision for the entire numerical analysis is set to at least 12, ensuring the maintenance of a precision of 12 digits during internal computations. Moreover,
the constanticity of energy first integral $H=E$, as given in~\eqref{eq:hamiltonian}, is used for the verification of the numerical integrations. We keep the relative and absolute errors up to~$10^{-11}$.

	Fig.~\ref{fig:spectrum} presents a  spectrum of   Lyapunov exponents for system~\eqref{eq:vh}, computed for constant values of the parameters
			\begin{equation}
			\begin{split}
			\label{eq:parametry}
		M&=2,\quad	m_1=2,\quad m_2=1,\quad a=5,\\
		 g&=1,\quad k_1=\frac{1}{10},\quad k_2=\frac{1}{4},
		\end{split}
		\end{equation}
	with the initial conditions
	\begin{equation}
		\begin{aligned}
			\label{eq:initial_cond}
&	\ell_0=1,&&\varphi_0=\frac{\pi}{2}, &&\vartheta_0\in (-\pi,\pi), \\ &
p_{\ell 0}=0.002,&& p_{\vartheta 0}=0.001,&& p_{\varphi 0}=0.001,
\end{aligned}
	\end{equation}
where $\vartheta_0$ is treated as the control parameter.  The considered system has six-dimensional phase space, therefore there are six Lyapunov exponents $\Lambda=\{\lambda,\lambda_2,\lambda_3,\lambda_4,\lambda_5,\lambda_6\}$,  where $\lambda\equiv \lambda_1$ is the largest Lyapunov exponent. 	Fig.~\ref{fig:spectrum} illustrates the impact of the initial swing angle $\vartheta_0$ on the system dynamics. In regions, where the Lyapunov exponents are larger than a numerical cut-off (typically  0.002 in our case), the separation is exponential, indicating hyperchaotic dynamics. Conversely,  when all Lyapunov exponents tend towards zero, the separation is slower than exponential, and thus the dynamics is regular (non-chaotic).  

Because the considered model is the Hamiltonian one, its Lyapunov exponents spectrum exhibits distinctive properties.  Firstly, 
the existence of the first integral, which is the conservation of the energy $H = E$, ensures that one pair of Lyapunov's exponents is zero~\cite{Pikovsky:16::}. Moreover, 
	the preservation of volume in phase space (Liouville’s theorem)  implies that the sum of all Lyapunov exponents is equal to zero~\cite{Vallejo:17::}. Finally,     due to the time reversibility in the Hamiltonian vector field~\eqref{eq:vh}, the Lyapunov exponents appear in additive inverse pairs.   As we are dealing with  three-degrees of freedom Hamiltonian system, the possible spectrum is given by $\Lambda=\{\lambda,\lambda_2,\lambda_3,-\lambda_3,-\lambda_2,-\lambda\}$, where $\lambda_3\approx 0$. In the considered case, the maximal value of $\lambda$ occurs in the neighborhood of the point $\vartheta_0=3.11$, where $\lambda\approx 0.058$. 
	
	Fig.~\ref{fig:lyap3da} presents a three-dimensional diagram of Lyapunov exponents  $(\lambda,\lambda_2,\lambda_3)$ as a function of initial swing angles $(\vartheta_0,\varphi_0)\in (\-\pi,\pi)$. On the right, the projections of $\lambda$ and $\lambda_2$  onto the  $(\vartheta_0,\varphi_0)$-plane, with color scales associated with the magnitudes of exponents. 
	 These colorful diagrams were obtained by numerically computing  Lyapunov exponents on a grid of $500\times 500$ values of $(\vartheta_0,\varphi_0)$ over the range $(-\pi,\pi)$.  It shows how the change of the initial swing angles $\vartheta_0$ and $\varphi_0$ of the pendulums (with almost initial velocities) affects the dynamics of the whole system. As expected, the diagram is quite symmetric about zero. Within it, we can observe the coexistence of regular and hyperchaotic dynamics,  depending on values of the control parameters $\vartheta_0$ and $\varphi_0$.
	For sufficiently small amplitudes of $(\vartheta_0,\varphi_0)$, the system performs regular (non-chaotic) oscillations. However, for larger values of the initial angles, hyperchaotic motion prevails, reaching its maximum intensity around the points $(\vartheta_0,\varphi_0) \approx \pm \pi$. 

	In Fig.~\ref{fig:lyap3da}, we can notice a  very good correspondence between $\lambda$ and $\lambda_2$, i.e., the regions with $\lambda$ and $\lambda_2$ larger than zero coincidence. Therefore, to specify values of $(\vartheta_0,\varphi_0)$ for which the motion of the system is hyperchaotic, we can limit ourselves to plotting $\lambda$ only. This is because if there exists an additional first integral inside the system, then it will be independent of initial conditions. 
Fig.~\ref{fig:lyap2d}  illustrates the Lyapunov diagrams for the largest exponent in the plane of the initial swinging angles $(\vartheta_0,\varphi_0)$ with increasing values of mass $M$. The color scale is proportional to the magnitude of $\lambda$. As we can notice, the situation becomes more complex. For $M=3$, the regular central part of the diagram decays and we observe the appearance of hyperchaotic behavior of the system even for very small values of the initial swing angles of the pendulums. In fact, at the central part of the Lyapunov diagram,  $\lambda$ reaches its maximal value. Further increments in the mass value $M$ increase the percentage value of the area in the diagram where the motion is hyperchaotic and the value of $\lambda$ is increasing as well. One can observe regular islands bounded by hyperchaotic regions.  Finally, for $M=9$, the entire region corresponding to the regular motion decays into global hyperchaos. For $M=9$, the largest Lyapunov exponent reaches its maximal value, up to $\lambda=0.61$.
	
	\begin{figure*}[htp]
		\centering                              
		\includegraphics[width=.3\linewidth]{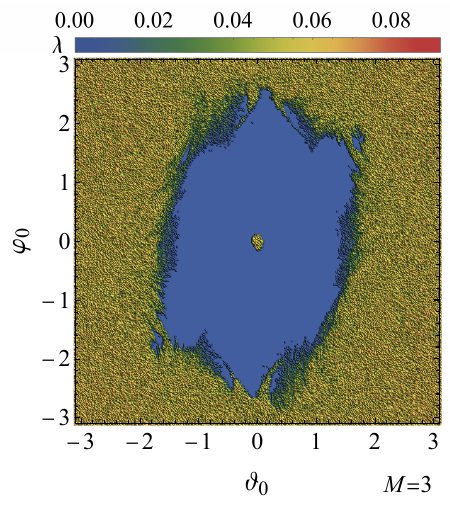}\hspace{-0.2cm}	
		\includegraphics[width=.3\linewidth]{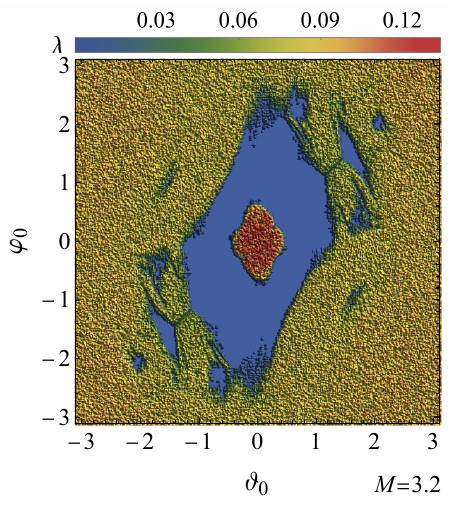}\hspace{-0.2cm}
			\includegraphics[width=.3\linewidth]{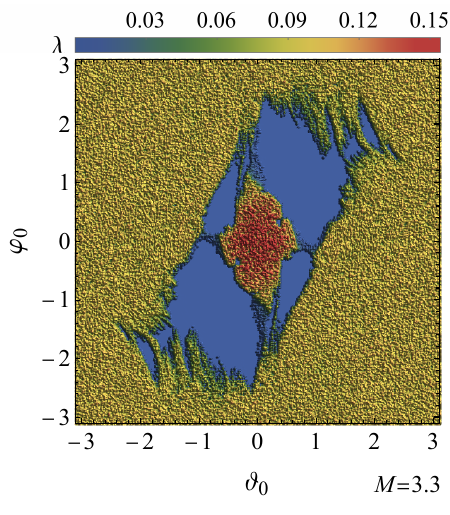}\\
		\includegraphics[width=.3\linewidth]{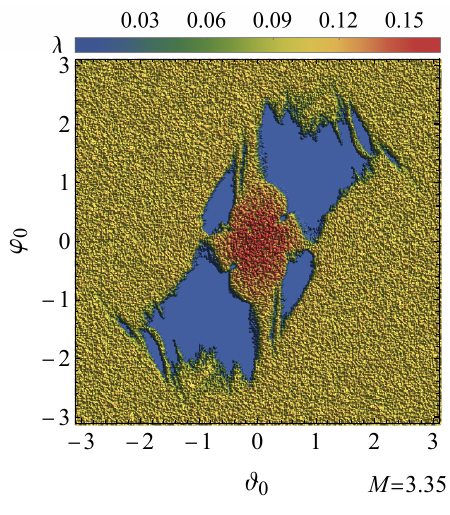}
		\includegraphics[width=.3\linewidth]{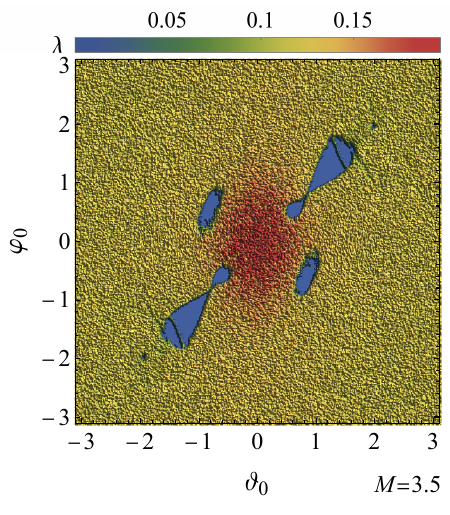}\hspace{-0.2cm}	
			\includegraphics[width=.3\linewidth]{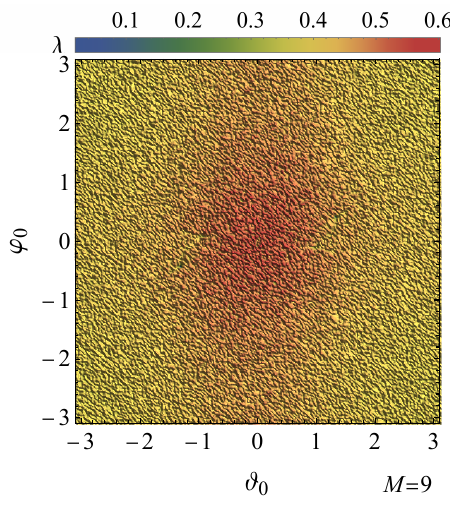}
		\caption{(Color online) The Lyapunov diagrams for the largest exponent in the plane of initial swinging angles $(\vartheta_0,\varphi_0)$ computed for varying values of  mass~$M$.   The colorful diagram was obtained by numerically computing  Lyapunov's exponents on a grid of $500\times 500$ values of $(\vartheta_0,\varphi_0)$ taken over the range $(-\pi,\pi)$ with the initial condition $(\ell_0=1,\,	p_{\ell 0}=0.002,\, p_{\vartheta 0}=0.001,\, p_{\varphi 0}=0.001)$. 
			The color scale is proportional to the magnitude of $\lambda$. The remaining parameters were chosen as: $m_1=2,\, m_2=1,\, a=5,\, g=1,\, k_1=0.1,\, k_2=0.25$.  Regions with $\lambda\approx 0$ correspond to regular (non-chaotic) dynamics, while regions with $\lambda>0$ are responsible for the hyperchaotic behavior of the system. As is evidenced by higher values of $M$, the regular regions divergence into hyperchaos. 
			\label{fig:lyap2d}}
	\end{figure*}
	\begin{figure*}[htp]
		\centering
		\subfigure[\,Global view]{
		\includegraphics[width=.4\linewidth]{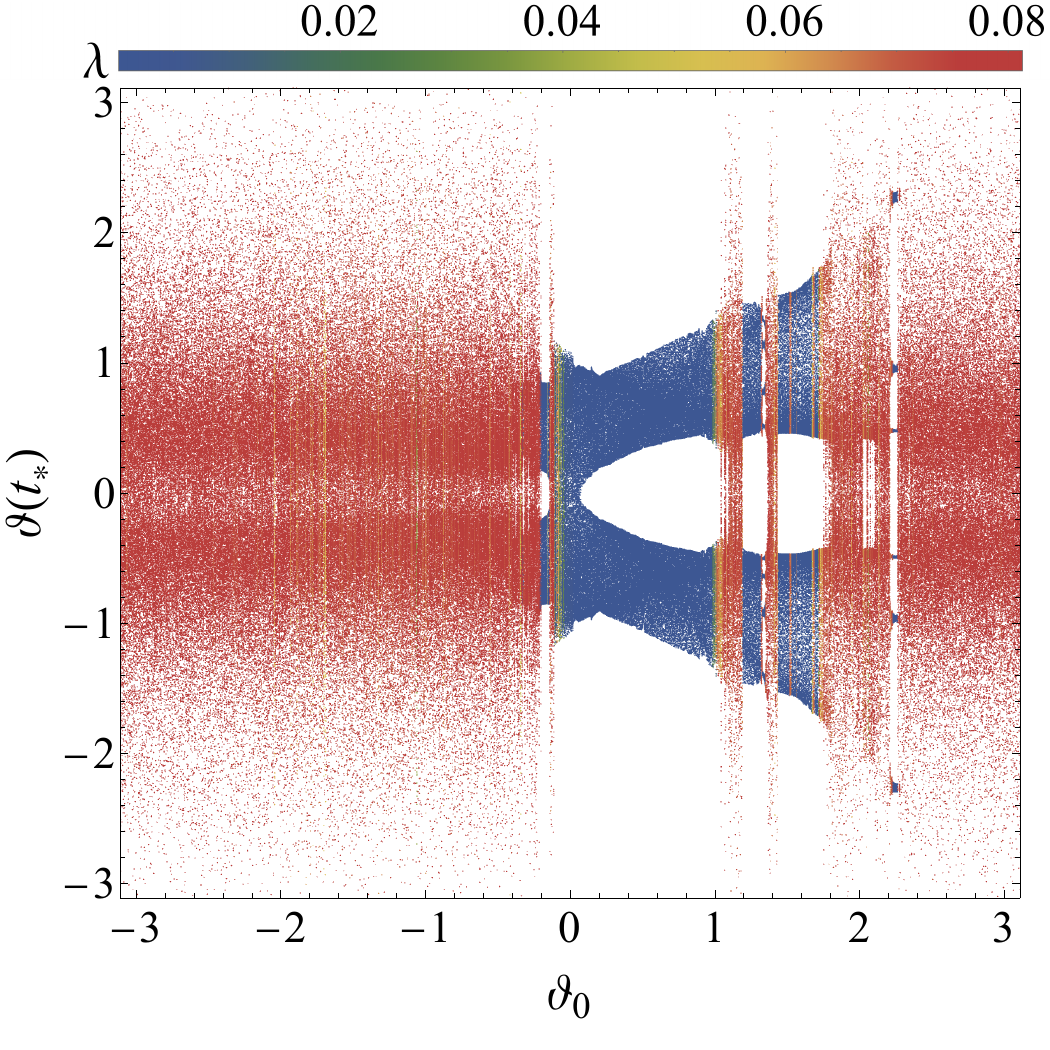}}
	\subfigure[\ Magnification presenting ,,periodic windows'' between chaotic layers]{
		\includegraphics[width=.4\linewidth]{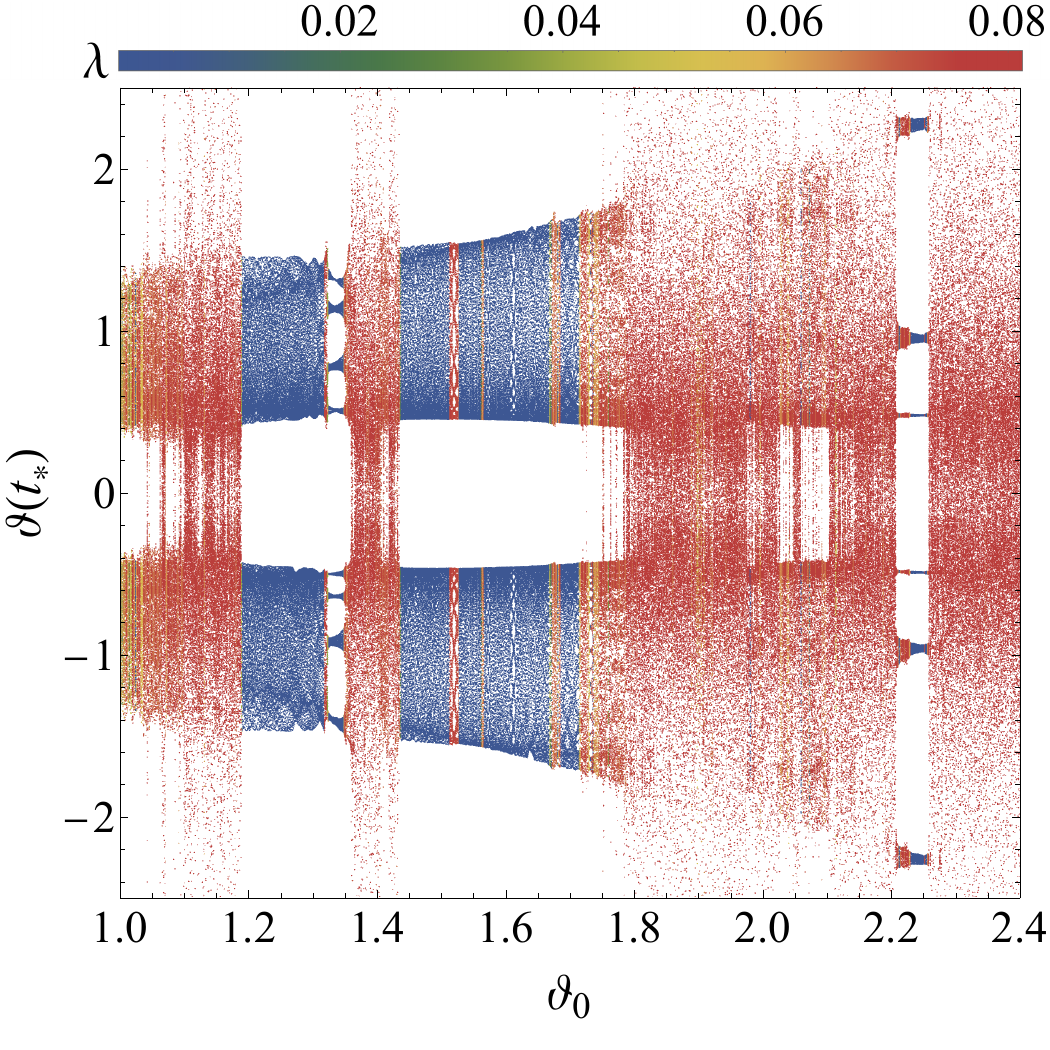}}
		\caption{(Color online) The phase-parametric diagram of system~\eqref{eq:vh} versus the initial swing angle $\vartheta_0$. Initial conditions and values of the parameters are taken from Fig.~\ref{fig:lyap2d} with $M=3.3$ and $\varphi_0=\pi/2$, while $\vartheta_0$ is treated as the control parameter.  Here, $\vartheta'(t_\star)= 0$ with  $\vartheta''(t_\star)< 0$,   for some~$t_\star$. 
			  The diagram is combined with the largest Lyapunov exponent $\lambda$. The color scale is proportional to the magnitude of $\lambda$. A very good agreement of the phase-parametric diagram with $\lambda$ is observed. The coexistence of periodic, quasi-periodic, and chaotic orbits together with ,,periodic windows'' between chaotic layers is visible.}
		\label{fig:lyap_bif}
	\end{figure*}
	
	\begin{figure*}[t]
		\centering \subfigure[{\,Periodic orbits $\vartheta_0=2.24, \varphi_0=\pi/2$}]
		{
			\includegraphics[width=.32\linewidth]{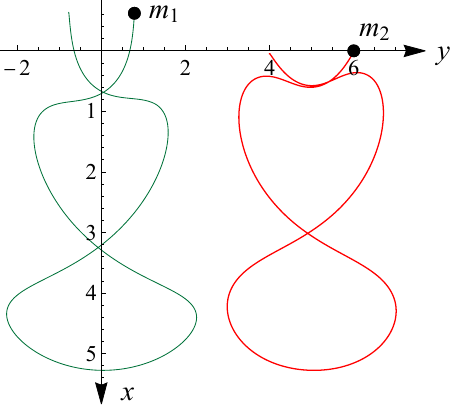}}
		\subfigure[{\, Quasi-periodic orbits $\vartheta_0=1.3, \varphi_0=\pi/2$}]
		{
			\includegraphics[width=.32\linewidth]{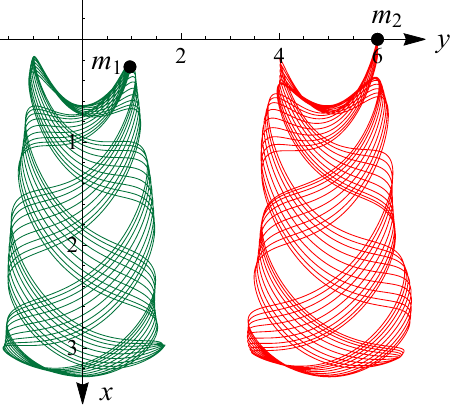}}
		\subfigure[\,  Chaotic  orbits $\vartheta_0=-3, \varphi_0=\pi/2$]{
			\includegraphics[width=.32\linewidth]{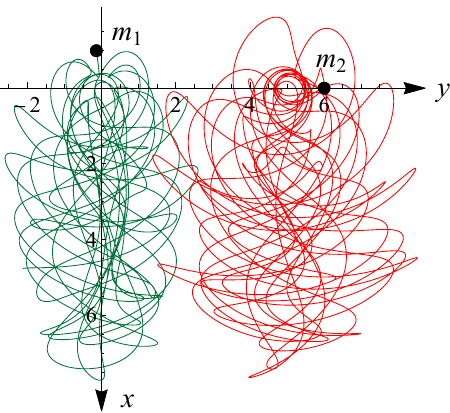}}
		\caption{(Color online) The periodic, quasi-periodic, and chaotic trajectories of coupled pendulums plotted in Cartesian plane. Dots $m_1$ and $m_2$ denote the initial swing angles $(\vartheta_0, \varphi_0)$. Respective values of $\vartheta_0, \varphi_0$ were taken from the phase-parametric diagram and corresponds well to the Lyapunov diagram with marked periodic orbits visible in Fig.~\ref{fig:lyap_period}. }
		\label{fig:phase}
	\end{figure*}
	\begin{figure}
		\centering
		\includegraphics[width=.8\linewidth]{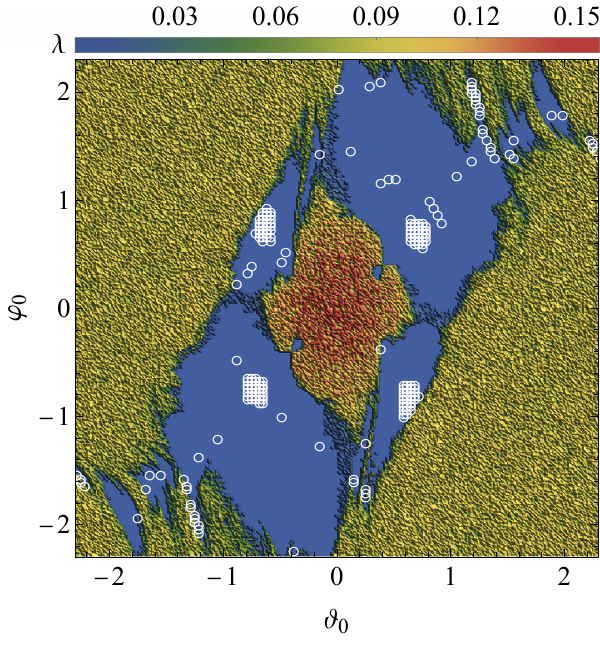}\\
				\includegraphics[width=.8\linewidth]{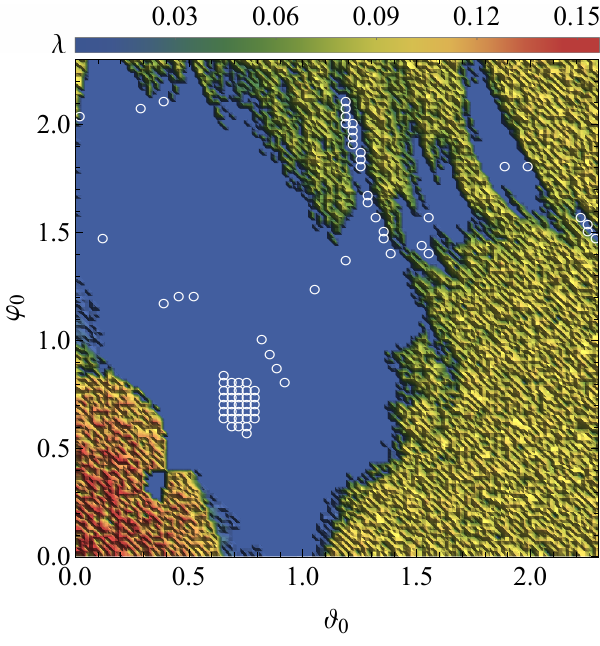}\\
		\caption{(Color online) Magnifications of the Lyapunov diagram, taken from Fig.~\ref{fig:lyap2d} for $M=3.3$, presenting the number of periodic orbits at $(\vartheta_0,\varphi_0)$-plane. Each depicted circle corresponds to values of the initial swing angles $(\vartheta_0,\varphi_0)$ for which the motion of the system is periodic.}
		\label{fig:lyap_period}
	\end{figure}

	\subsection{Phase-parametric diagram}
	The two-parameter diagrams of the Lyapunov exponents, visible in Fig.~\ref{fig:lyap2d}, provide quantitative insight into the dynamics of the considered model by specifying intervals of the initial swing angles $(\vartheta_0,\varphi_0)$ for which the motion is either regular or hyperchaotic.   However, because our model is the Hamiltonian system, we were not able to deduce from these figures whether the observed regular patterns  (where $\lambda\approx 0$) correspond to periodic or quasi-periodic motion. 
To make such a distinction, the construction of phase-parametric (bifurcation) diagrams is helpful. Briefly speaking, a phase-parametric diagram shows periodic orbits and their frequency ratios, routes to the chaos, and the periodic windows between chaotic regions by plotting the dependence of a chosen state variable as a function of a certain control parameter~\cite{Stachowiak:15::,Szuminski:20::}.
	
Fig.~\ref{fig:lyap_bif} illustrates the phase-parametric diagram of the system computed for a one-parameter family of initial conditions taken from Fig.~\ref{fig:lyap2d} with $M=3.3$ and $\varphi_0=\pi/2$, where $\vartheta_0\in (-\pi,\pi)$ is treated as the control parameter.
 In this calculated phase-parametric diagram, we display the dependence of the maximal values (amplitudes) of $\vartheta(t)$ on the initial swing angle  $\vartheta_0\in (-\pi,\pi)$. That is, for a given initial condition, we consecutively integrate equations of motion~\eqref{eq:vh}, and we build the diagram by collecting points $\vartheta(t_\star)$  when $\vartheta'(t_\star)=0$ and $\vartheta''(t_\star)<0$. As a result, we obtain a pattern on the plane, which can be easily interpreted. To enhance the analysis, we overlay the obtained phase-parametric diagram with the largest Lyapunov exponent~$\lambda$. The color scale is associated with the magnitude of~$\lambda$.
In Fig.~\ref{fig:lyap_bif}, we observe a very good agreement of the phase-parametric diagram with~$\lambda$.  Indeed, for $\lambda>0$, the phase-parametric diagram illustrates the complex dynamics of the system visible in terms of random-looking points.  Conversely, for $\lambda\approx $ the shape of the phase-parametric diagram is regular.  However, what was not visible on the Lyapunov diagrams, 
inside the regular regimes of the phase-parametric diagram,    quasi-periodic and periodic orbits are distinguishable.  For better understanding,   we show in Fig.~\ref{fig:phase} exemplary periodic, quasi-periodic, and chaotic trajectories plotted in the Cartesian plane. The initial conditions were drawn from the phase-parametric diagram. 
Dots $m_1$ and $m_2$ state for initial amplitudes $(\vartheta_0,\varphi_0)$ of the pendulums. 
 Despite the general hyperchaotic nature of the system,  we can still find values of $(\vartheta_0,\varphi_0)$ for which the motion is periodic. It is especially visible in the magnification of the phase-parametric diagram taken over the range $\vartheta_0\in(1,2.4)$, where the periodic gaps between chaotic layers are visible.

All in all, the  Lyapunov exponents spectrum is a very useful tool to measure the complexity and the strength of chaos in the system dynamics, while phase-parametric diagrams are effective in identifying periodic orbits and their characteristics.
Therefore, let us combine these two methods more systematically.
	We do this in the following way.
	For the given values of parameters~\eqref{eq:parametry} 
and initial conditions~\eqref{eq:initial_cond},
	we build a grid of $500\times 500$ values of $(\vartheta_0,\varphi_0)$ over the range $(-\pi,\pi)$. Then,   for each initial condition, we compute the Lyapunov exponents. If $\lambda>0$,  the corresponding initial conditions are excluded from the set. As a result, we obtain a grid  $B$ of $n$ initial conditions for which the motion of the system is non-chaotic. The second step is to numerically integrate  equations of motion~\eqref{eq:vh} for $(\vartheta_0,\varphi_0)\in B$, and to built  diagrams by collecting points  $\vartheta(t_\star)$ when $\vartheta'(t_\star)=0$ and $\vartheta''(t_\star)<0$ for a certain $t_\star$.  This process yields  $n$ lists with intersecting points $\vartheta(t_\star)$.  Within each list,  we look for the scheme of repeating values of $\vartheta(t_\star)$ in a specific order. 
 In this way, a rough but effective distinction between periodic and quasi-periodic motion is possible. 
	
	Fig.~\ref{fig:lyap_period} displays the Lyapunov exponents diagram on the $(\vartheta_0,\varphi_0)$-plane with marked dots for which the motion of the system is periodic.   This figure provides a comprehensive view of the system dynamics, allowing for the identification of chaotic, quasi-periodic, and periodic regions. It complements the analysis, making it exhaustive. For example,
	the phase-parametric diagram visible in Fig.~\ref{fig:lyap_bif}, corresponds to
	Fig.~\ref{fig:lyap_period} with the
	chosen initial
	$\varphi_0=\pi/2$ and $\vartheta\in(-\pi,\pi$). We observe a  very good agreement between these two plots.  Indeed, along the line $\varphi_0=\pi/2$,  we can find in Fig.~\ref{fig:lyap_period} two periodic circles, i.e., at  $\vartheta_0\approx 1.33$ and $\vartheta_0\approx 2.24$, which is suitable with the magnification of the phase parametric diagram visible in Fig.~\ref{fig:lyap_bif}(b). 
	\subsection{Invariant manifold and the Poincar\'e cross-sections}

	\begin{figure}[t!]
	\begin{center}
		\resizebox{85mm}{!}{\begingroup%
		  \makeatletter%
		\providecommand\color[2][]{%
			\errmessage{(Inkscape) Color is used for the text in Inkscape, but the package 'color.sty' is not loaded}%
			\renewcommand\color[2][]{}%
		}%
		\providecommand\transparent[1]{%
			\errmessage{(Inkscape) Transparency is used (non-zero) for the text in Inkscape, but the package 'transparent.sty' is not loaded}%
			\renewcommand\transparent[1]{}%
		}%
		\providecommand\rotatebox[2]{#2}%
		\newcommand*\fsize{\dimexpr\f@size pt\relax}%
		\newcommand*\lineheight[1]{\fontsize{\fsize}{#1\fsize}\selectfont}%
		\ifx\svgwidth\undefined%
		\setlength{\unitlength}{248.83760708bp}%
		\ifx\svgscale\undefined%
		\relax%
		\else%
		\setlength{\unitlength}{\unitlength * \real{\svgscale}}%
		\fi%
		\else%
		\setlength{\unitlength}{\svgwidth}%
		\fi%
		\global\let\svgwidth\undefined%
		\global\let\svgscale\undefined%
		\makeatother%
  \begin{picture}(1,0.61995447)%
    \put(0,0){\includegraphics[width=\unitlength,page=1]{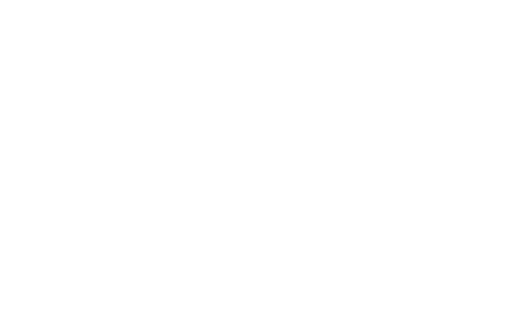}}%
    \put(0.11018461,0.25442396){\color[rgb]{0,0,0}\makebox(0,0)[lt]{\smash{$m_1$}}}%
    \put(0,0){\includegraphics[width=\unitlength,page=2]{obr2.pdf}}%
    \put(0.58617018,0.22599443){\color[rgb]{0,0,0}\makebox(0,0)[lt]{\smash{$m_2$}}}%
    \put(0.92133212,0.1477025){\color[rgb]{0,0,0}\makebox(0,0)[lt]{\smash{$k_1$}}}%
    \put(0.28112126,0.48963447){\color[rgb]{0,0,0}\makebox(0,0)[lt]{\smash{$a$}}}%
    \put(0.35437556,0.38293292){\color[rgb]{0,0,0}\makebox(0,0)[lt]{\smash{$g$}}}%
    \put(0.04883005,0.38526653){\color[rgb]{0,0,0}\makebox(0,0)[lt]{\smash{$\vartheta$}}}%
    \put(0.13026143,0.36943391){\color[rgb]{0,0,0}\makebox(0,0)[lt]{\smash{$\ell$}}}%
    \put(0.57292431,0.35868698){\color[rgb]{0,0,0}\makebox(0,0)[lt]{\smash{$\ell$}}}%
    \put(0.01597875,0.43860272){\color[rgb]{0,0,0}\makebox(0,0)[lt]{\smash{$x$}}}%
    \put(0.06544636,0.49129834){\color[rgb]{0,0,0}\makebox(0,0)[lt]{\smash{$y$}}}%
    \put(0,0){\includegraphics[width=\unitlength,page=3]{obr2.pdf}}%
    \put(0.90873989,0.37476143){\color[rgb]{0,0,0}\makebox(0,0)[lt]{\smash{$X$}}}%
    \put(0.94628668,0.27157102){\color[rgb]{0,0,0}\makebox(0,0)[lt]{\smash{$M$}}}%
    \put(0,0){\includegraphics[width=\unitlength,page=4]{obr2.pdf}}%
			\end{picture}%
			\endgroup}
\caption{ (Color online) Geometry of the variable-length coupled pendulums moving in
				the gravitational and Hooke’s potentials. Here $M$ and $m_1,m_2$ are the masses linked by an inextensible string of lengths $l_1=l+a$ (red) and  $l_2=l$ (green). Masses $m_1$ and $m_2$ can swing, whereas $M$ is constrained to move solely in the vertical direction. The Hamiltonian function that describes the model is defined in~\eqref{eq:hamiltonian}. \label{fig:1}}
	\end{center}
	\end{figure}
	The Poincar\'e cross-section method is an essential tool in qualitative dynamics analysis, especially for Hamiltonian systems of two degrees of freedom. This method is based on intersections of phase curves with a properly chosen surface of the section in a three-dimensional hypersurface defined by a constant energy level. As a result, we obtain a pattern on the section plane, which is easy to visualize and interpret~\cite{Stachowiak:15::,Szuminski:22::}.

Because the considered model is the Hamiltonian system of three degrees of freedom it is complicated to deduce any useful information from the Poincar\'e sections.	However,  if we look closer at equations of motion~\eqref{eq:vh}, we can notice that for $k_2=0$, the system admits an invariant manifold
	\begin{equation}
		\label{eq:invariant0}
		\scM=\left\{(\ell,p_\ell,\vartheta,p_\vartheta,\varphi,p_\varphi)\in \R^6\, \big{|} \, \varphi=p_\varphi=0\right\}.
	\end{equation}
	\begin{figure*}
	\centering \subfigure[{\,$E=1.63$, chaotic central part of the section. }]
	{
		\includegraphics[width=.85	\linewidth]{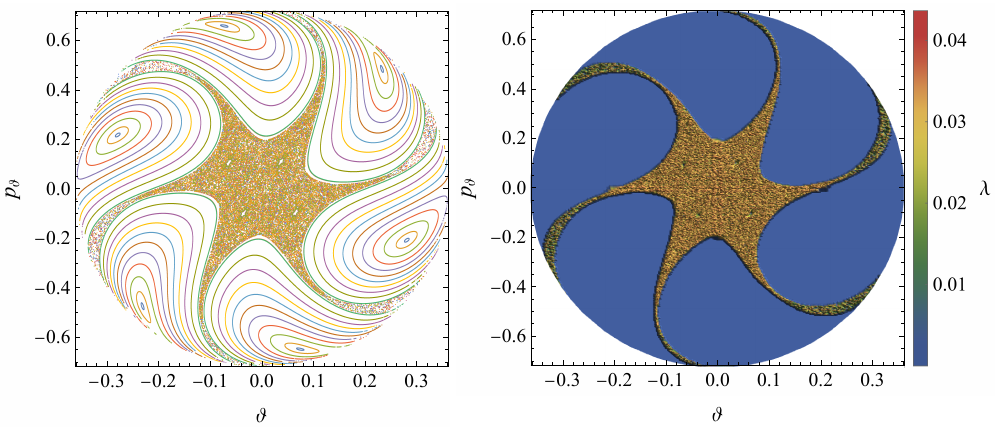}	}
	\subfigure[{\, $E=1.7$, the rise of regular islands between chaotic layers. }]
	{
		\includegraphics[width=.85 	\linewidth]{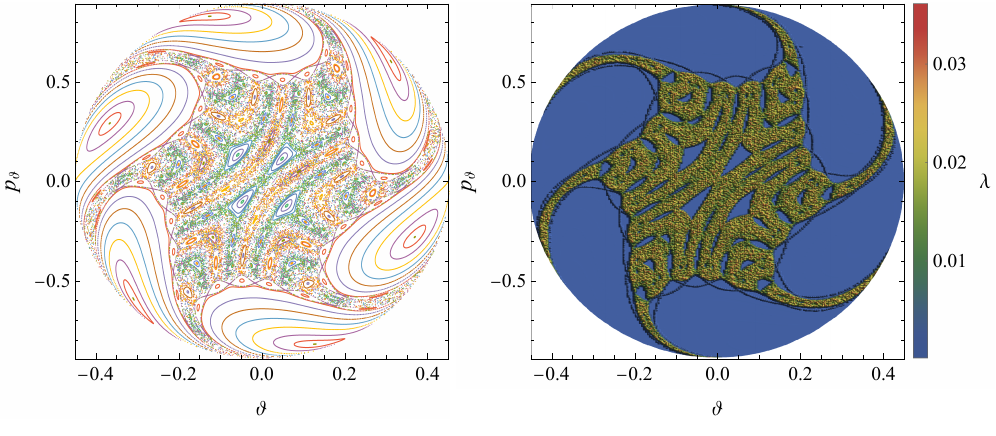}	}
	
	\caption{(Colour online) The Poincar\'e section of system~\eqref{eq:vha} and their corresponding Lyapunov diagrams made for $M=4,\, m_1=2,\, m_2=1,\,  k_1=1,\,  g=1$, with gradually increasing values of energy $E$. The cross-section plane was specified as $\ell=1$ with direction $p_\ell>0$. Each color at the Poincar\'e plane corresponds to distinct initial conditions, while in the Lyapunov diagram, the color scale is proportional to the magnitude of~$\lambda$. }
	\label{fig:poin1}
\end{figure*}

	Hamiltonian~\eqref{eq:hamiltonian},  constrained to manifold $\scM$, reduces to a system of two degrees of freedom, as illustrated in Fig.~\ref{fig:2}. Therefore, under the initial condition $\varphi_0=\dot \varphi_0=0$, the dynamics of the original Hamiltonian vector field~\eqref{eq:vh} is equivalent to the following system:
	\begin{equation}
		\begin{cases}
			\label{eq:vha}
			\dot \ell=\dfrac{p_\ell}{M+m_1+m_2},\\[0.3cm]
			\dot p_\ell=\dfrac{p_\vartheta^2}{m_1\ell^3}-g(M-m_2-m_1\cos\vartheta)
			-k_1\ell,
			\\[0.3cm] \dot\vartheta=\dfrac{p_\vartheta}{m_1\ell^2},\\[0.2cm] \dot p_{\vartheta}=-m_1 g \ell \sin\vartheta,
		\end{cases}
	\end{equation}
	with the Hamiltonian first integral
	\begin{equation}
		\label{eq:hamiltonian_red}
		\begin{split}
			&	\widetilde{H}=\dfrac{1}{2}\left(\dfrac{p_\ell^2}{M+m_1+m_2}+\dfrac{p_\vartheta^2}{m_1\ell^2}\right)
			\\ & \hspace{2cm}+g\ell(M-m_2-m_1\cos\vartheta)+\frac{1}{2}k_1 \ell^2.\end{split}
	\end{equation}
	As the evolution of the reduced system takes place in four-dimensional phase space, the Poincar\'e sections method can be effectively adopted.    
	The main idea of the Poincar\'e cross-sections is very simple. We consider a three-dimensional surface (in our case  $\ell_0=1$) in the phase space which is traversed by all trajectories,  together with the energy a constant energy-level $\mathcal{M}_E=\{\widetilde{H}(\ell_0,p_\ell,\vartheta,p_\vartheta)=E\} $  which is also three dimensional. In general, set $\mathcal{M}_E$ is not connected, that is it consists of several separated parts. In the considered case it has two connected components $\mathcal{M}_E^+$ and $\mathcal{M}_E^-$.  
 They are distinguished in the following ways. 
 We choose $(\ell,\vartheta,p_\vartheta)$ as coordinates on  the level $\mathcal{M}_E$. For a given point $(\ell,\vartheta,p_\vartheta)\in\mathcal{M}_E$, 
we have two choices $p_{\ell\pm}=p_{\ell\pm}(E,\ell,\vartheta,p_\vartheta)$,
which correspond to components $\mathcal{M}_E^\pm$. The cross-section plane $\ell=\ell_0$ cuts both components $\mathcal{M}_E^\pm$. This cut is two-dimensional, and we take $(\vartheta,p_\vartheta)$ as coordinates on it. In figures, we present the part of this cut contained in the component $\mathcal{M}_E^+$. As a result, we obtain a pattern in the plane, which is easy to visualize and interpret. In summary, if the motion is periodic, the trajectory passes through the plane only in a finite number of intersections. If the motion is quasi-periodic a single orbit fills densely a finite number of continuous loops. A chaotic trajectory intersects the plane in scattered, random-looking points. 
	\begin{figure*}
		\centering
		\subfigure[{\, $E=2.07$, the beauty of  the coexistence of periodic, quasi-periodic and chaotic orbits. }]
		{
			\includegraphics[width=.85	\linewidth]{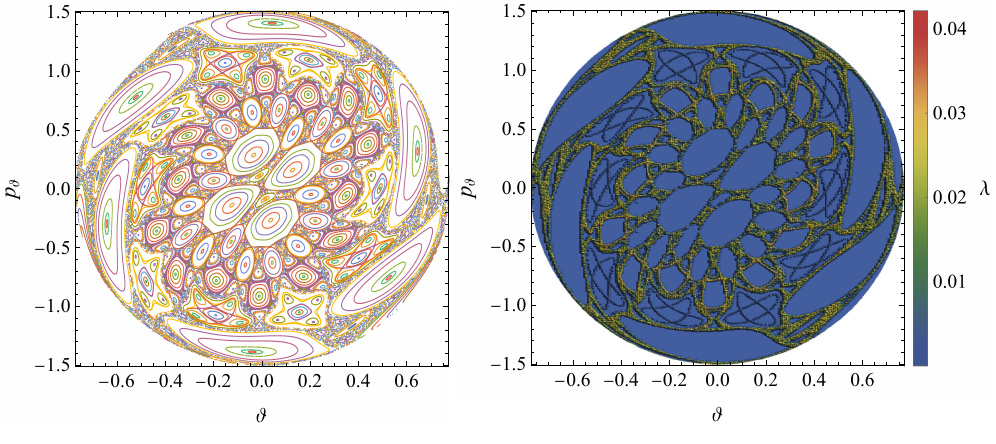}	}
		\subfigure[{\, $E=3.82$, regular and chaotic orbits for higher values of the energy. }]
		{
			\includegraphics[width=.85	\linewidth]{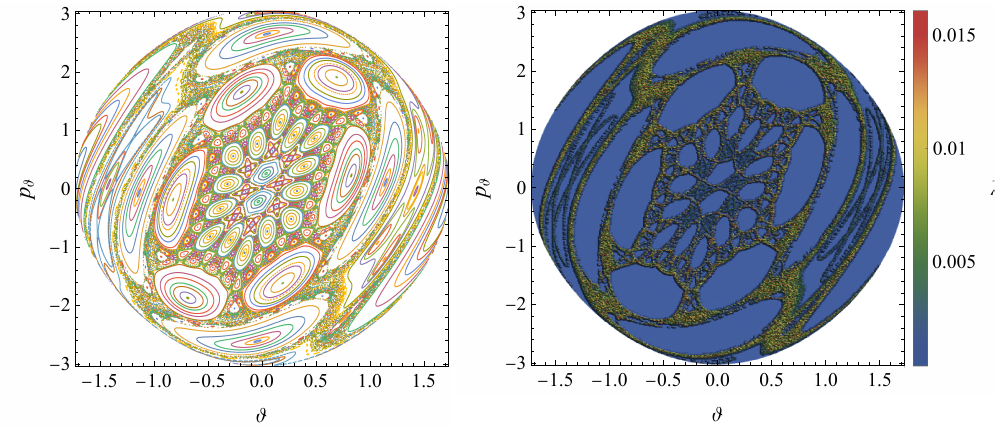}	}
		\caption{(Color online) The Poincar\'e sections of system~\eqref{eq:vha} and their corresponding Lyapunov diagrams made for $M=4,\,  m_1=2,\, m_2=1,\,  k_1=1,\, g=1$, with gradually increasing values of energy $E$. The cross-section plane was specified as $\ell=1$ with direction $p_\ell>0$.  Each color at the Poincar\'e plane corresponds to distinct initial conditions, while in the Lyapunov diagram, the color scale is proportional to the magnitude of $\lambda$.  The plots indicate the beautiful coexistence of periodic, quasi-periodic, and chaotic orbits of the system. }
		\label{fig:poin2}
	\end{figure*}
	
	\begin{figure*}[t]
		\centering
		\includegraphics[width=.43\linewidth]{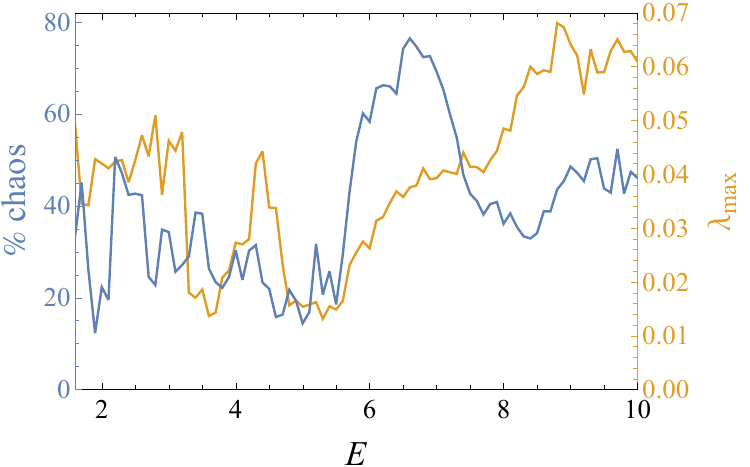}
		\includegraphics[width=.43\linewidth]{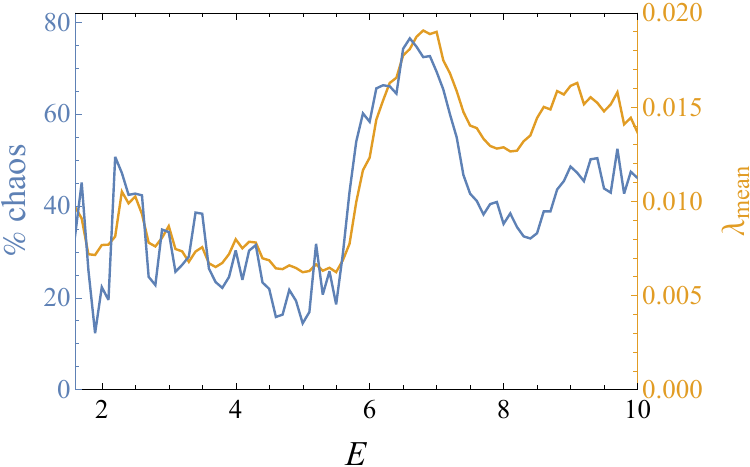}
		\caption{(Color online) The percentage of chaos versus maximal (left) and mean~(right) values of the largest Lyapunov exponent $\lambda$ in phase space as a function of energy.}
		\label{fig:probability}
	\end{figure*}
	Figs.~\ref{fig:poin1}-\ref{fig:poin2} depict the Poincar\'e sections and their corresponding two parameter Lyapunov diagrams of the system~\eqref{eq:vha}, constructed for constant parameters:
	\begin{eqnarray}
		M=4,\quad m_1=2,\quad m_2=1, \quad k_1=1, \quad g=1.
	\end{eqnarray}
	for gradually increased values of the energy. Each color in the Poincar\'e sections corresponds to a distinct initial condition, while in the Lyapunov exponents diagrams, the color scale is proportional to the values of the largest exponent $\lambda$. At first sight, we observe a  very good correlation between the Poincar\'e sections and their corresponding Lyapunov exponents diagrams.  The areas filled by scattered points in the Poincar\'e plane have non-zero values of~$\lambda$, which allows us to quantify the presence of chaos in the system.

	In Fig.~\ref{fig:poin1}(a), we present the first pair of the Poincar\'e section and the Lyapunov exponents diagram computed for the same value of the energy $E=1.63$.   In the central part of the Poincar\'e plane, we observe the prominent signs of chaotic behavior of the system manifested as scattered, random-looking points.
This observation is somewhat unusual because,  for Hamiltonian systems with energies close to energy minimum,  the Poincar\'e sections typically show regular patterns with shapely-elegant quasi-periodic loops, as illustrated in~\cite{Stachowiak:06::,Szuminski:22::,Szuminski:20::,Szuminski:20b::}.  In our case, however, the chaotic region splits the Poincar\'e section into six smaller areas, each with a particular periodic solution bounded by quasi-periodic orbits.

As the energy value increases, the behavior of the system deviates significantly from what is typically observed in Hamiltonian systems. Fig.~\ref{fig:poin1}(b) shows the Poincar\'e section for a slightly higher value of the energy, for  $E=1.7$. We can notice that instead of decaying the successive invariant tori to prone the system to be more ergodic, we observe the appearance of new stable periodic solutions in the central part of the plane. This is somehow even more evident when we further increase the values of the energy. The Poincar\'e section, visible in Fig.~\ref{fig:poin2}, presents the beautiful coexistence of periodic, quasi-periodic and chaotic orbits. Indeed, almost the entire figure is covered by neckless formations corresponding to periodic motion.  Moreover,  while not being visible via the Poincar\'e section, we can observe in the Lyapunov exponents diagrams the chaotic folds. This is attributed to the computation of Lyapunov exponents for 
	$500\times 500$ values of $(\vartheta_0,\varphi_0)$ allowing us to encounter initial conditions responsible for "weak" chaotic orbits. Thus, the Lyapunov exponents diagram serves as a complementary tool to the Poincar\'e sections, providing insights into chaotic dynamics that may not be apparent from the latter, especially when constructed for a much smaller grid of initial conditions for practical and visual reasons.	
	
Lyapunov exponents can be also used as an estimator for the percentage area of the Poincar\'e plane occupied by chaotic motion. Indeed, for a large number of initial conditions uniformly distributed in the available area of the Poincar\'e plane, we repeatedly compute the largest Lyapunov exponent. Then, we take the ratio of the number of points with a Lyapunov exponent different from zero (typically larger than $\lambda_\text{min}=0.002)$ to the total amount of points in the sample. We repeat the entire procedure by increasing the values of the energy $E$.
Similarly,  we can calculate the maximum and mean values of the largest Lyapunov exponent for a given energy.  The results of these computations are presented in Fig.~\ref{fig:probability}. These plots confirm the earlier observation during the description of the Poincar\'e sections that there is no typical transition from regular, almost integrable dynamics at low energies to ergodic dynamics at higher energy values. Instead, we observe alternating increases and decreases in the percentage area of chaos as a function of energy.  Unexpectedly, for very high values of  $E$, the percentage of the area responsible for regular orbits at the Poincar\'e plane remains prominent.  The average value of the Lyapunov exponent $\lambda_\text{mean}$ exhibits similar behavior, while the maximal value of the Lyapunov exponent $\lambda_\text{max}$ increases proportionally with the energy.

\section{The nonintegrability proof \label{sec:2}} The performed numerical analysis
 reveals complex and mostly hyperchaotic system dynamics. However, such analysis was made for fixed values of the
parameters. 
For other parameter sets, the results can be completely different, and, in some cases, the system can possess a first integral and even be integrable, precluding its chaotic behavior. It is, therefore, natural to perform a comprehensive integrability analysis of the Hamiltonian~\eqref{eq:hamiltonian}. For this purpose, we employ the Morales--Ramis theory~\cite{Morales:99::}. This theory is based on the analysis of the differential Galois group of variational equations obtained by the linearization of equations of motion along a certain particular solution.  The main theorem of this theory states that if the Hamiltonian system is integrable in the Liouville sense, then the identity component of the differential Galois group of variational equations must be Abelian.   For a more instructive introduction to this theory, interested readers can refer to~\cite{Morales:99::,Morales:00::,Morales:01::}.
	
Below we formulate the main theorem of this paper.
\begin{theorem}
		\label{th:main}
Let $M, m_1, m_2$ and $a$ are positive parameters and $g\neq 0$. If  the
variable length coupled pendulum system governed by
Hamiltonian~\eqref{eq:hamiltonian} is integrable in a class of functions meromorphic in coordinates and momenta, then
		\begin{eqnarray}
			\label{eq:warunki}
			k_1= k_2= 0,\quad \text{and}\quad \frac{M}{m_1+m_2}=1+\frac{4}{p^2+p-4},
		\end{eqnarray}
		for $p\in \N,\ p\geq 2$.
	\end{theorem}
	
	\begin{proof}
		System~\eqref{eq:vh} possesses the following invariant manifold
		\begin{equation}
			\label{eq:invariant}
			\scN=\left\{\left(\ell,p_\ell,\vartheta,p_\vartheta,\varphi,p_\varphi\right)\in {\C}^6\,\big{|}\, \vartheta=\varphi=0=p_\vartheta=p_\varphi\right\}.
		\end{equation}
		Restricting the right-hand sides of~\eqref{eq:vh} to $	\scN$, we obtain a Hamiltonian system of one degree of freedom
		\begin{equation}
			\label{eq:H0_vh0}
			\dot \ell=\dfrac{p_\ell}{M+m_1+m_2}, \quad \dot p_\ell=-g(M-m_1-m_2)-k_1 \ell,
		\end{equation}
		with the Hamiltonian
		\begin{equation}
			H=\frac{p_\ell^2}{2(M+m_1+m_2)}+g(M-m_1-m_2)\ell+\frac{1}{2}k_1 \ell^2.
		\end{equation}
		Eqs.~\eqref{eq:H0_vh0} can be rewritten as a one-second-order non-homogeneous Newton's equation
		\begin{equation}
			\label{eq:vh02}
			\ddot \ell+\left(\frac{k_1}{M+m_1+m_2}\right)\ell=g\left(1-\frac{2M}{M+m_1+m_2}\right).
		\end{equation}
 Hence, it can be easily solved using the simple shift in the variable. We have a whole family of particular solutions defined by the energy first integral
		\begin{equation}
			\label{eq:calka_energi}
			E=\left(\frac{M+m_1+m_2}{2}\right)\dot \ell^2+g(M-m_1-m_2)\ell+\frac{1}{2}k_1 \ell^2.
		\end{equation}
		Solving Eq.~\eqref{eq:vh02} and taking into the account the energy first integral~\eqref{eq:calka_energi}, we get the particular solution  $\vvarphi(t)=(\ell(t),p_\ell(t),0,0,0,0)$. The form of $\ell(t)$ depends on whether the spring $k_1$ is equal to zero or not.

		 For  $k_1\neq 0$, equation~\eqref{eq:vh02} is the second-order differential equation with the harmonic oscillator-like solution		
		\begin{equation}
			\begin{split}
				\label{eq:rozw}
				\ell(t)=A \cos\left[ \omega_1(t-t_0)\right]+\delta,
			\end{split}
		\end{equation}
		where $A$ is an amplitude of oscillations, while  $\omega_1$ is a natural frequency, defined by
		\begin{equation}
			\label{eq:parametry1}\begin{split}
				&A=\sqrt{\frac{2E}{k_1}+\delta^2},\quad \omega_1=\sqrt{\frac{k_1}{M+m_1+m_2}},\\
				&\delta=-\frac{g(M-m_1-m_2)}{k_1}.
		\end{split}\end{equation}
		Otherwise,  when $k_1=0$,  equation~\eqref{eq:vh02}  translates to the equation describing the motion of the classical Atwood's machine. Integrating twice, we obtain
		\begin{equation}
			\label{eq:rrr}
			\ell(t)=-\frac{1}{2}at^2+v_0 t+l_0,\quad a=\left(\frac{M - m_1 - m_2}{M + m_1 + m_2}\right)g.
		\end{equation}
		where $v_0, l_0$ are constants of the integrations related to an initial velocity and an initial distance.

		Let $\vX=(L,P_L,\Theta, P_\Theta, \Phi, P_\Phi)^T$ denotes  variations of $\vx=(\ell,p_\ell,\vartheta,p_\vartheta,\varphi,p_\varphi)^T$, then the variational equations of system~\eqref{eq:vh}, along the particular solution $\vvarphi(t)$, are as follows 
		
		\begin{equation}
			\begin{split}
				\dot \vX=\vA\cdot \vX,\quad \text{where}\quad \vA= \pder{\vv_H}{\vx}(\varphi(t)),
			\end{split}
		\end{equation}
		where $\vv_H$ states for the Hamiltonian vector field~\eqref{eq:vh}. The  explicit form of matrix $\vA$ is 
		\begin{equation*}
			\begin{split}
				&\vA=\\ &\begin{pmatrix}
					0&\frac{1}{M+m_1+m_2}&0&0&0&0\\
					-k_1&0&0&0&0&0\\
					0&0&0&\frac{1}{m_1\ell^2}&0&0\\
					0&0&-gm_1 \ell-k_2\ell^2&0& k_2 \ell^2&0\\
					0&0&0&0&0&\frac{1}{m_2 \ell^2}\\
					0&0& k_2\ell^2&0&-g m_2 \ell-k_2 \ell^2&0
				\end{pmatrix}
			\end{split}
		\end{equation*}
		As we can notice, this variational system splits  into two subsystems: the normal variational equations for the variables $(\Theta, 
				P_\Theta,
				\Phi,
				P_\Phi )^{T}$ and tangential equations for $(L,P_L)^{T}$. Since the tangential subsystem is trivially solvable,  for further consideration we take the normal part, which takes the form		\begin{equation}
			\begin{pmatrix}
				\dot \Theta \\
				\dot P_\Theta\\
				\dot \Phi \\
				\dot P_\Phi 
			\end{pmatrix}
			=
			\begin{pmatrix}
				0&\frac{1}{m_1\ell^2}&0&0\\
				-gm_1 \ell-k_2\ell^2&0& k_2 \ell^2&0\\
				0&0&0&\frac{1}{m_2 \ell^2}\\
				k_2\ell^2&0&-g m_2 \ell-k_2 \ell^2&0
			\end{pmatrix}
			\begin{pmatrix}
				\Theta \\
				P_\Theta\\
				\Phi \\
				P_\Phi 
			\end{pmatrix}
		\end{equation}
		This system can be rewritten as a one-fourth-order differential equation for variable~$\Theta$. Its explicit form is given by
		\begin{equation}
			\begin{split}
				\label{eq:fourth}
				&0=\ddddot \Theta(t)+4 \left(\frac{\dot \ell}{\ell}\right) \dddot \Theta(t)+ \left(\frac{2(g+2 \ddot \ell)}{\ell}+\omega_2^2\right) \ddot \Theta(t) +\\
				&2\left(\frac{(g-\ddot \ell)\dot \ell}{\ell^2}+\omega_2^2\frac{\dot \ell}{\ell}+\frac{\dddot \ell}{ \ell}\right)\dot \Theta(t)+g\left(\frac{g-\ddot \ell}{\ell^2}+\frac{\omega_2^2}{ \ell}\right) \Theta(t),\\
			\end{split}
		\end{equation}
		where $\omega_2$ is the reduced frequency given by
		\[\omega_2=\sqrt{\frac{(m_1+m_2)}{m_1m_2}k_2}.\]

Let us remark on this point. The integrability analysis via the differential Galois theory often involves differential systems wherein the normal variational equations can be transformed into an independent subsystem of second-order differential equations. Next, through an appropriate change of independent variables, these equations can be rationalized and their differential Galois groups can be effectively studied employing the Kovacic algorithm~\cite{Kovacic:86::}.
This algorithm  classifies the
possible types solutions of second-order differential equations with rational coefficients.
Unfortunately, there is no equivalent of the Kovacic algorithm for linear
differential equations with rational coefficients of higher orders, although many
partial results are known~\cite{Singer:95::,Ulmer:03::}. Perhaps, the most
compressive results can be found in the recent work~\cite{Combot:18b::}, where the
authors present the equivalent of the Kovacic algorithm for symplectic
differential operators of dimension four.
		
Fortunately, the obtained fourth-order variational equation~\eqref{eq:fourth} exhibits a nice property --- it can be factorized. We state the following.

		\begin{lemma}
			Let us define the differential operators
			\begin{equation}
				\begin{split}
					\label{eq:diff_operator}
						\mathscr{L}_1=D_t^2+2 \left(\frac{\dot \ell}{\ell}\right)D_t+\frac{g}{\ell},\qquad \mathscr{L}_2=\mathscr{L}_1+\omega_2^2,
				\end{split}
			\end{equation}
			where $(D_t=\rmd/\rmd t)$. 
			The differential operators  $\mathscr{L}_1, \mathscr{L}_2$ commute and their actions
			\begin{equation}
				\label{eq:operatory}
				\mathscr{L}_1 [ \mathscr{L}_2 \, \Theta(t)]=0, \quad \text{or}\quad \mathscr{L}_2[\mathscr{L}_1\, \Theta(t)]=0,
			\end{equation}
			coincide with fourth-order differential equation~\eqref{eq:fourth}.
		\end{lemma}
		\begin{proof}
			It is easy to show that operators $\mathscr{L}_1, \mathscr{L}_2$ commute, i.e.,
			\begin{equation}
				\begin{split}
					&\mathscr{L}_1\mathscr{L}_2=\mathscr{L}_1\left[\mathscr{L}_1+\omega_2^2\right]=\mathscr{L}_1\mathscr{L}_1+\omega_2^2\mathscr{L}_1=\mathscr{L}_2\mathscr{L}_1.
				\end{split}
			\end{equation}
		Explicit computations of~\eqref{eq:operatory} are straightforward but lengthy, so we leave them to the interested reader. 
		\end{proof}
		
		Since the fourth-order variational equation~\eqref{eq:fourth} factories, and the operators~\eqref{eq:operatory} commute,  we can, without loss of generality, independently study the differential Galois groups of  $\mathscr{L}_1$ and $L_{2} $~\cite{Combot:18b::}. 
		Moreover, to prove the nonintegrability of our system, it is enough to show that the identity component of the differential Galois group of either of these operators is not Abelian.
We achieve this by employing the classical Kovacic algorithm of dimension two.
		
Given that the Kovacic algorithm was constructed for reduced rational second-order differential equations, we have to perform appropriate changes of variables to differential operators~\eqref{eq:diff_operator}.  In our case,  these changes of variables depend on whether the constant $k_1$ is zero or not. Therefore, we analyze these two cases independently. 
		
		\subsection{Case with $\,k_1\neq 0$}
		We start with the following change of the independent variable 
		\begin{equation}
			\label{eq:rational_change}
			t\longrightarrow z:=1+\frac{\ell(t)}{A-\delta}=\frac{A}{A-\delta}\left(1+ \cos\left[ \omega_1(t-t_0)\right]\right),
		\end{equation}
		Taking into account the transformation rules for derivatives
		\begin{equation}
			\label{eq:pochodne}
			D_t=\dot \ell \, D_z,\quad D_t^2=\ddot \ell\, D_z+\dot \ell^2\,D_z^2,
		\end{equation}
		we perform the rationalization of operators~\eqref{eq:diff_operator}. Their explicit forms are given by
		\begin{equation}
			\begin{split}
				\label{L1L2z}
				\mathscr{L}_1=D_z^2+P(z)\, D_z+Q(z),\qquad \mathscr{L}_2=\mathscr{L}_1-\frac{\omega}{z(z-\alpha)},
			\end{split}
		\end{equation}
		where 
		\begin{equation}
			\begin{split}
				P=\frac{1}{2z}+\frac{2}{z-1}+\frac{1}{2(z-\alpha)},\quad Q=-\frac{\Omega}{z(z-1)(z-\alpha)}
			\end{split}.
		\end{equation}
		The new dimensionless parameters  are defined as
		\begin{equation}
			\label{eq: parametry}
			\Omega:=\frac{g}{(A-\delta)\omega_1^2},\qquad \omega:=\frac{\omega_2^2}{\omega_1^2},\qquad \alpha=\frac{2A}{A-\delta}.
		\end{equation}
		\begin{corollary}
			Rationalization~\eqref{eq:rational_change}  of the fourth-order  variational equation~\eqref{eq:fourth}, give rise to the fourth order differential operator, which is the least common left multiple  (LCLM) of~\eqref{L1L2z}, namely
			\begin{equation}
				\mathscr{L}=\operatorname{LCLM}(\mathscr{L}_1,\mathscr{L}_2)=\operatorname{LCLM}(\mathscr{L}_2,\mathscr{L}_1),
			\end{equation}
		\end{corollary}
		Next, we perform the following change of the dependent variable
		\begin{equation}
			\label{eq:reducedtran}
			\mathscr{L}_{i}[\Theta(z)]=\mathscr{L}_{i}\left[w(z)\exp\left(-\frac{1}{2}\int_{z_0}^zP(z')dz'\right)\right]=\mathscr{D}_{i}[w(z)],
		\end{equation}
		which transform operators~\eqref{L1L2z} into their reduced forms
		\begin{equation}
			\begin{split}
				\label{eq:Heun}
				\mathscr{D}_1=D_z^2-R_1(z),\qquad 
				\mathscr{D}_2=\mathscr{D}_1-R_2(z).
			\end{split}
		\end{equation}
		where \begin{equation}
			\begin{split}
				\label{R1R2}
				&R_1=\frac{12z^3+4(1-5\alpha+4\Omega)z^2+\alpha(5\alpha-4-16\Omega)z+3\alpha^2}{16z^2(z-1)(z-\alpha)^2}, \\
				&R_2=R_1-\frac{\omega}{z(z-\alpha)}.
			\end{split}
		\end{equation} 
To prove the nonintegrability of the system, governed by Hamiltonian~\eqref{eq:hamiltonian}, it is sufficient to show that the identity component of the differential Galois group G of at least one of the operators~\eqref{eq:Heun} is not Abelian. To check these possibilities, we introduce theorems, which describe all possible types of G  and relate them to the forms of solutions of~\eqref{eq:Heun}.   Following Kovacic's approach,  we state
		\begin{theorem}[Kovacic]
			\label{lem:m_1_a}
			Let $\operatorname{G}$ be the differential Galois group of the differential operator
			\begin{equation}
				\label{eq:reduced}
					\mathscr{D}=D_z^2-R(z),\qquad R(z)\in \Q(z).
			\end{equation}
			Then, one of the four cases can occur.
			\begin{enumerate}
				\item $\operatorname{G}$ is conjugate to a subgroup of triangular group
				\[
				\scT=\left\{\begin{pmatrix}
					a&b\\
					0 &a^{-1}
				\end{pmatrix}\vert a\in \C^*, b\in \C\right\}.
				\]
				and equation $\mathscr{D}[w(z)]=0$ has an exponential solution  $w=P\exp[\int \xi]$, $P\in \C[z],\ \xi \in \C(z)$.
				\item $\operatorname{G}$ is conjugated with a subgroup of
				\[
				\scD^\dagger=\left\{\begin{pmatrix}
					c&0\\ 0&c^{-1}
				\end{pmatrix}\vert c\in \C^*\right\} \cup 
				\left\{\begin{pmatrix}
					0&c\\ -c^{-1}&0
				\end{pmatrix}\vert c\in \C^*\right\};
				\]
				in this case equation  $\mathscr{D}[w(z)]=0$ has a solution of the form $w=\exp[\int \xi]$, where $\xi$ is algebraic function of degree $2$.
				\item $\operatorname{G}$ is  finite and all solutions of $\mathscr{D}[w(z)]=0$ are algebraic.
				\item $\operatorname{G}$ is $\operatorname{SL}(2,\C)$ and  equation $\mathscr{D}[w(z)]=0$ has no Liouvillian solution.
			\end{enumerate}
		\end{theorem}
		\begin{remark}
			Let us write $R(z)\in \C(z)$ in the form
			\[
			R(z)=\frac{p(z)}{q(z)},\qquad p(z),q(z)\in \C[z].
			\]
			The roots of $q$ are the poles of $R$. Let \[
			\Sigma=\Sigma'\cup \{\infty\},\qquad \Sigma'=\left\{c\in \C \,| \, q(c)=0\right\}
			\] be the finite set of poles of  $R$ in the complex plane with infinity as well.
			The order of the pole $c\in \Sigma'$, which we denote simply by $\operatorname{o}(c)$,  is the
			multiplicity of  $c$ as a root of $q$, and  the order of  infinity is
			$\operatorname{o}(\infty)=\operatorname{deg}(q)-\operatorname{deg}(p)$.
		\end{remark}
		\begin{theorem}[Kovacic]
			\label{lem_m_1_b}
			The following conditions are necessary for the respective cases given in Theorem~\ref{lem:m_1_a}.
			\begin{enumerate}
				\item Every pole  $c\in \Sigma'$ must have even order or else have order $1$. Moreover, the order $\operatorname{o}(\infty)$  must be even or else greater than $2$.
				\item The set $\Sigma'$ contains  at least one pole $c$ that either has odd order greater than
				$2$ or else has order $2$.
				\item The order $\operatorname{o}(c)\leq 2$  and the order  $\operatorname{o}(\infty)\geq 2$. If the partial fraction expansion of $R$ is
				\[
				R(z)=\sum_i\frac{a_i}{(z-c_i)^2}+\sum_j\frac{b_j}{z-d_j},
				\]
				then $\Delta_i=\sqrt{1+4a_i}\in \Q$ for each $i$, $\sum_j b_j=0$ and if
				\[
				G=\sum_ia_i+\sum_j b_jd_j,
				\]
				then $\sqrt{1+4G}\in \Q$.
			\end{enumerate}
		\end{theorem}
		
		Let us return to our case. The operators~\eqref{eq:Heun} belong to the generalized Heun's family, with four regular singular points located at $\Sigma=\{0,1,\alpha, \infty\}$.   To avoid the confluence of singularities, we assume $\alpha\neq 0$ and $\alpha \neq 1$. 
	Singularities $z=0$ and $z=\alpha$ are poles with orders $\operatorname{o}(0)=2=o(\alpha)$, while $\operatorname{o}(1)=1$.  The degree of infinity is $\operatorname{o}(\infty)=2$.  
		Thus, taking into account the character of these singularities, we can deduce that necessary conditions for all cases given in Theorem~\ref{lem_m_1_b} are satisfied. Hence, according to  Theorem~\ref{lem:m_1_a}, the differential Galois group of $\mathscr{D}_1$ and $\mathscr{D}_2$ can be reducible, finite, dihedral or $\operatorname{SL}(2, \C)$.  To analyze these four distinct cases, we use the Kovacic algorithm.
		
		We start by computing the  Laurent series expansions of  $R_1(z)$ and $R_2(z)$  about the singularities  $c_i \in \Sigma$ with the order $\operatorname{o}(c_i)=2$, i.e., $\{0,\alpha, \infty\}$. The expressions  are as follows
		\begin{enumerate}
			\item around $z=0$
			\[R_1(z)=-\frac{3}{16z^2}+\ldots, \quad  R_2(z)=-\frac{3}{16z^2}+\ldots,
			\]
			\item  around $z=\alpha$ 
			\[R_1(z)=-\frac{3}{16(z-\alpha)^2}+\ldots, \quad R_2(z)=-\frac{3}{16(z-\alpha)^2}+\ldots,
			\]
			\item around $z=\infty$
			\[ 
			R_1(z)=\frac{3}{4z^2}+\ldots,\quad R_2(z)=\frac{3+4\omega}{4z^2}+\ldots .
			\]\end{enumerate}
	The analysis of the differential Galois group of the operator $\mathscr{D}_2$ is considerably more complicated, mainly because the residue of $R_2(z)$ at infinity depends on the value of $\omega$. However, as mentioned earlier, 
		it is sufficient to show that the identity component of the differential Galois group of either $\mathscr{D}_1$ or $\mathscr{D}_2$ is not Abelian. Hence,  we restrict ourselves to the analysis of $\mathscr{D}_1$, due to the simpler characteristic exponent of $R_1(z)$ at infinity.

		\begin{lemma} The differential Galois group of operator $\mathscr{D}_1$~\eqref{eq:Heun} is $\operatorname{SL}(2,\C)$.
		\end{lemma}
		\begin{proof}
			\textbf{Case 1.} By the first case of the  algorithm, for  singularities $c\in\{0,\alpha,\infty\}$ with $\operatorname{o}(c)=2$, we compute 
			\begin{equation}
				\alpha_c^\pm=\frac{1}{2}\pm\frac{1}{2}\sqrt{1+4 a_c},\end{equation}
			where $a_c$ are coefficients of the Laurent series expansions of $R_1(z)$ about  $\{0,\alpha, \infty\}$, i.e.,
			\begin{equation}
				\label{eq:ac}
				a_c=\left\{-\frac{3}{16},-\frac{3}{16},\frac{3}{4}\right\}.
			\end{equation}		
			As the singularity  $z=1$  is the pole with order $\operatorname{o}(1)=1$, we set $\alpha_1^\pm=1$.
		Following the algorithm, we introduce the axially sets $E_c=\{\alpha_c^+,\alpha_c^-\}$,
			\begin{equation}
				E_0=E_\alpha=\left\{ \frac{3}{4},\frac{1}{4}\right\},\quad E_1=\{1,1\},\quad E_\infty=\left\{\frac{3}{2},-\frac{1}{2}\right\},
			\end{equation}
			Next, we calculate  the Cartesian product $E=E_0\times E_1\times E_\alpha\times E_\infty$, and to ensure integrability, we have to consider only those permutations  $e_c=(e_0,e_1,e_\alpha,e_\infty)$ that yield a non-negative integer of
			\begin{equation}
				\label{eq:de}
				d(e)=e_\infty-e_0-e_1-e_\alpha\in\N\cup\{0\}.
			\end{equation}
			In our case, there exists only one distinct element $e_c\in E$ satisfying this condition, namely
			\begin{equation}
				e=\left\{\frac{1}{4},1,\frac{1}{4},\frac{3}{2}\right\},\quad \text{with}\quad d(e_c)=0.
			\end{equation}
			Now we pass to the third step of the Kovacic algorithm.
			We look for a polynomial $P(z)\neq 0$ of degree $d(e)$, such that it is a solution of the following differential equation
			\begin{equation}
				\label{eq:srr}
		P''+2 w P'+(w'+w^2-R_1(z))P=0,
			\end{equation}
	where 	$R_1(z)$ is defined in~\eqref{R1R2} and
			\begin{equation}
				\label{omega}
				w(z)=\sum_{c\in \Sigma'}\frac{e_c}{z-c}=\frac{1}{4z}+\frac{1}{z-1}+\frac{1}{4(z-\alpha)}.
			\end{equation} In the considered case, we have $P=1$, so
Eq.~\eqref{eq:srr} simplifies considerably
			\begin{equation}
				\label{eq:finale}
			w'+w^2=R_1(z),
			\end{equation}
Direct computations show that this equality cannot be satisfied for arbitrary $z$ except $\Omega=0$. However, setting $\Omega$ to zero implies $g=0$, which is in contradiction with our assumption from Theorem~\ref{th:main}. Hence,  we conclude that Kovacic’s algorithm does not find an exponential solution of the form $w(z)=P\exp[\int \xi]$, where $P\in \C[z]$, and $\xi \in \C(z)$. 
			
			\textbf{Case 2.}  For   singularities $c\in\{0,\alpha,\infty\}$ with $\operatorname{o}(c)=2$, we define sets of exponents
			\begin{equation}
				E_c=\{2, 2\pm 2\sqrt{1+4 a_c}\}\cap \Z,
			\end{equation}
			where  the coefficients $a_c$ are~\eqref{eq:ac}. For $z=1$, we have $\operatorname{o}(1)=1$, so we define $E_1=\{4\}$. Hence,  the explicit forms of the auxiliary sets $E_c$,  are given by
			\begin{equation}
				\label{eq:setsE}
				E_0=E_\alpha=\{1,2,3\},\quad E_1=\{4\},\quad E_\infty=\{-2,2,6\}.
			\end{equation}
			Next, we look for elements $e_c=(e_0,e_1,e_\alpha,e_\infty)\in (E_0\times E_1\times E_\alpha\times E_\infty)$, for which 
			\begin{equation}
				d(e)=e_\infty-e_0-e_1-e_\alpha\in \N_{\operatorname{even}}\cup\{0\}.
			\end{equation}
			As it turns out, we have only one element satisfying this condition, namely
			\begin{equation}
				e_c=\{1,4,1,6\},\quad \text{with}\quad d(e_c)=0.
			\end{equation}
			Then, for the above set, we can construct a rational function 
			\begin{equation}
				u(z)=\frac{1}{2}\sum_{c\in \Sigma'}\frac{e_c}{z-c}=\frac{1}{2}\left(\frac{1}{z}+\frac{4}{z-1}+\frac{1}{z-\alpha}\right),\end{equation}
			and we need to find a monic polynomial $P$ od degree $d(e_c)$, such that
			\begin{equation}
\begin{split}
					&P'''+3uP''+(3u^2+3u'-4 R_1(z))P'\\ &+(u''+3 u u'+u^3-4 u R_1(z)-2R_1'(z))=0.
	\end{split}
			\end{equation}
			Since $d(e_c)=0$, we set $P=1$. Thus,  the existence of $P$ translates to  checking whether $u(z)$ satisfies the following differential equation
			\begin{equation}
			(u''+3 u u'+u^3-4 u R_1(z))=2R_1'(z),
			\end{equation}
		where 	$R_1(z)$ is defined in~\eqref{R1R2}
			The above differential equation is fulfilled only if $\Omega=0$, which is in contradiction to our assumption. Hence, the second case of the Kovacic algorithm is not satisfied as well. 
			
			\textbf{Case 3.} In the third case of the algorithm, the auxiliary sets  for $c\in\{0,\alpha,\infty\}$ with $\operatorname{o}(c)=2$,
			are defined as follows
			\begin{equation}
				E_c=\left\{6\pm k\sqrt{1+4a_c} \, |\, k=0, 1,\ldots,  6\right\} \cap \Z,
			\end{equation}
			In this case, for  $c=1$ with $\operatorname{o}(1)=1$,  we have $E_1=\{12\}$. Hence, the explicit forms of the sets $E_c$ with $c\in \Sigma$, are given by
			\begin{equation}
				\label{eq:E3}
				\begin{split}
					&E_0=E_\alpha=\{3, 4, 5, 6, 7, 8, 9\},\quad E_1=\{12\},\\ & E_\infty=\{-6, -4, -2, 0, 2, 4, 6, 8, 10, 12, 14, 16, 18\}.
				\end{split}
			\end{equation}
			Next, we select from the product $E=(E_0\times E_1\times E_\alpha\times E_\infty)$ these elements $e=(e_0,e_1,e_\alpha,e_\infty)$, for which quantity $d(e)$ defined previously in~\eqref{eq:de} is satisfied. As it turns out, among~$637$ combinations there is no element $e$, for which $d(e)$ is a non-negative integer. Thus, the algorithm stops, and there are no solutions in this case as well.

			From the direct application of Kovacic's algorithm, we conclude that the differential operator $\mathscr{D}_1$~\eqref{eq:Heun} is not solvable. Since the three first cases of Theorem~\ref{lem:m_1_a}  do not hold the fourth case is automatically satisfied.  The differential Galois group of the reduced operator $\mathscr{D}_1$~\eqref{eq:Heun}  is $\operatorname{SL}(2, \C)$ with non-Abelian identity component. 
		\end{proof}
		
		\subsection{Case with $k_1=0$}
		For $k_1=0$ the particular solution, along which we compute the variational equations, has the form~\eqref{eq:rrr}. Therefore, in order to rationalize operators~\eqref{eq:diff_operator}, we perform the following linear change of  variable
		\begin{equation}
			\label{eq:rational_change2}
			t\longrightarrow z:=1-\frac{g(M-m_1-m_2)}{h}\ell(t),
		\end{equation}
		After this rationalization operators $\mathscr{L}_1$ and $\mathscr{L}_2$  are as follows
		\begin{equation}
			\label{eq:L1L2rat}
			\mathscr{L}_1=D_z^2+p(z) D_z+q(z),\qquad \mathscr{L}_2=\mathscr{L}_1+\frac{\epsilon}{z},
		\end{equation} with
		\begin{equation}
			p(z)=\frac{1}{2z}+\frac{2}{z-1},\quad q(z)=\frac{1+\lambda}{2(1-\lambda)z(z-1)}.
		\end{equation}
		Here $\lambda$ and $\epsilon$ are dimensionless parameters defined by
		\begin{equation}
			\lambda:=\frac{M}{m_1+m_2},\qquad \epsilon:=\frac{h\, \omega_2^2}{2a^2(M+m_1+m_2)}.
		\end{equation}
		Operator $\mathscr{L}_1$ is the Gauss hypergeometric differential operator with three regular singular points $\{0,1,\infty\}$. Differences of the exponents $e_1,e_2,e_3$ at  the respective  singular points are as follows
		\begin{equation}
			\begin{split}
				\label{eq:wykladniki}
				&e_1=\frac{1}{2},\qquad e_2=1,\qquad e_3=\frac{\sqrt{17\lambda-1}}{2\sqrt{\lambda-1}}.
			\end{split}
		\end{equation}
		If Hamiltonian~\eqref{eq:hamiltonian} for $k_1=0$ is integrable in the sense of Liouville, then the identity component of the differential Galois group of $\mathscr{L}_1$ must be Abelian. So, in particular, it is solvable. The necessary conditions for the solvability of the Gauss hypergeometric differential equation are well known due to Kimura's theorem~\cite{Kimura:69::}, see also the Appendix.   We state the following 
		\begin{lemma} If the differential Galois group of $\mathscr{L}_1$~\eqref{eq:L1L2rat}  has a solvable identity component, then
			\begin{equation}
				\label{eq:warunki2}
				\lambda=\frac{M}{m_1+m_2}=1+\frac{4}{p^2+p-4},\quad p\in \N,\quad p\geq 2.
			\end{equation}
		\end{lemma}
		\begin{proof}
			The proof consists of the direct application of the Kimura theorem to the obtained differential operator $\mathscr{L}_1$.  Due to the fixed differences of the exponents~\eqref{eq:wykladniki}, the calculations are straightforward.
		\end{proof}
		
		As we observe, there is a wide range of values of $\lambda$ for which there is no integrability obstacle. Therefore, we proceed with the analysis of the differential Galois group of the second operator $\mathscr{L}_2$~\eqref{eq:L1L2rat}.  For $h\omega_2\neq 0$, the operator $\mathscr{L}_2$ is the confluent Heun differential operator with two regular singularities $\Sigma'=\{0,1\}$ and one irregular at $\infty$. To employ the Kovacic algorithm,  we must express $\mathscr{L}_2[\Theta(z)]$, using the change of variable~\eqref{eq:reducedtran}, in the reduced form~\eqref{eq:reduced}.  The obtained equation $\mathscr{D}_2[w(z)]=0$, leads to the following differential operator
		\begin{equation}
			\label{eq:DD2}
			\mathscr{D}_2=D_z^2-R(z),\quad R(z)=-\frac{3}{16z^2}+\frac{p(p+1)}{4z(z-1)}-\frac{\epsilon}{z}.
		\end{equation}
		Point  $z=0$ is a pole with $\operatorname{o}(0)=2$, while  $z=1$ has $\operatorname{o}(1)=1$.  The degree of infinity is $\operatorname{o}(\infty)=1$.  Hence, taking into account the characteristic of the exponents at singularities $\Sigma=\{0,1,\infty\}$, we conclude that the differential Galois group of $\mathscr{D}_2$ cannot be reducible or finite because Case 1 and Case 3 of Theorem~\ref{lem_m_1_b} do not hold.  Differential Galois group of $\mathscr{D}_2$ can be only dihedral or
		$\operatorname{SL}(2,\C)$. To check the first possibility we apply the second case of the  Kovacic algorithm. We state the following.
		\begin{figure*}
		\centering
			\includegraphics[width=.85	\linewidth]{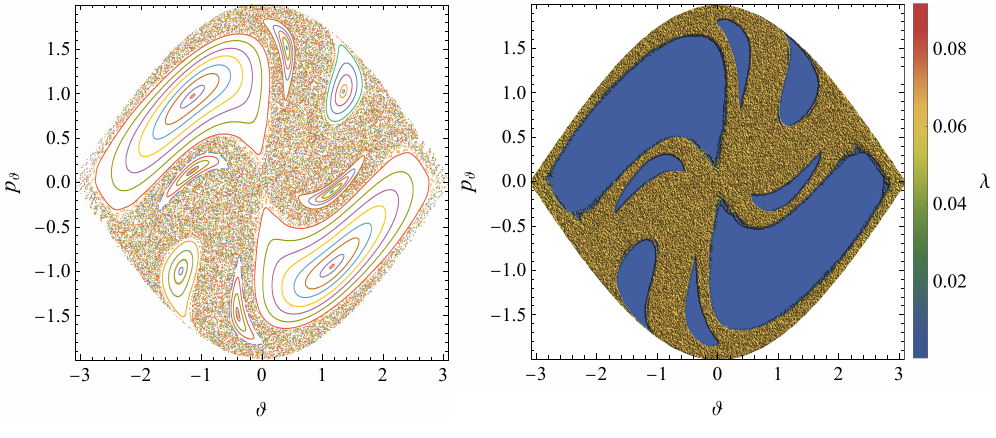}\\
				\includegraphics[width=.85	\linewidth]{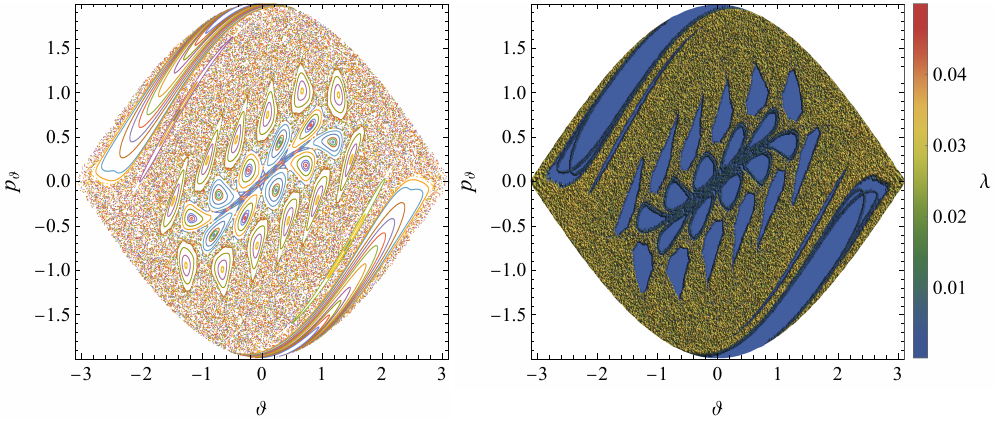}
			\caption{(Color online) The Poincar\'e sections of system~\eqref{eq:vh} restricted to the invariant manifold~\eqref{eq:invariant0}. We put $k_1=k_2=0$, and $m_1=m_2=1$, while $M=6$ and $M=3$ are taken to satisfy the necessary integrability condition~\eqref{eq:warunki}.   The cross-section plane is $\ell=1$ with direction $p_\ell>0$. The energy levels were chosen as  $E=E_0+2$, where $E_0$ is the energy minimum. The plots indicate the chaotic behavior of the system precluding its integrability.}
			\label{fig:poink1k20}
		\end{figure*}
		\begin{lemma} The differential Galois group of operator $\mathscr{D}_2$~\eqref{eq:DD2} is $\operatorname{SL}(2,\C)$.
		\end{lemma}
		\begin{proof}
			Following the second case of the Kovacic algorithm, For the respective singularities $\Sigma=\{0,1,\infty\}$ with degrees $\operatorname{o}(0)=2,\, \operatorname{o}(1)=1$ and $\operatorname{o}(\infty)=1$, we introduce the auxiliary sets
			\begin{equation}
				\label{eq:E1E2E3}
				E_0=\{1,2,3\},\qquad E_1=\{4\},\qquad E_\infty=\{1\}.
			\end{equation}
			Next, we compute the Cartesian product $E=(E_0\times E_1\times E_\infty)$, which gives only three possible combinations
			\begin{equation}
				E=\{\{1, 4, 1\},\{2, 4, 1\}, \{3, 4, 1\}\}.
			\end{equation}
			It is clear that there is no element $e_c=(e_0,e_1,e_\infty)\in E$, for which condition~\eqref{eq:de} holds.   As the set of positive $d(e_c)$ is empty the algorithm has stopped.  Therefore, the second case of the algorithm cannot occur, which implies that only the fourth case is possible, i.e., $\operatorname{G} = \operatorname{SL}(2,\C)$, and operator  $\mathscr{D}_2$~\eqref{eq:DD2}  has no Liouvillian solution.  
					\end{proof}
Based on our analysis,  we conclude that except in the case when both spring constants $k_1,k_2$ are zero, 	the identity component of the differential Galois group of the original fourth-order variational equation~\eqref{eq:fourth} is not Abelian. This implies that the variable-length coupled pendulum system governed by Hamiltonian~\eqref{eq:hamiltonian} is not integrable in a class of functions meromorphic in coordinates and momenta. This ends the proof. 
			\end{proof}		
		
	As we have shown, for $k_1=k_2=0$ and when the condition of the mass ratio~\eqref{eq:warunki2} is fulfilled, the necessary integrability conditions are satisfied and therefore the system is suspected to be integrable.  In Fig.~\ref{fig:poink1k20}, we present the Poincar\'e sections and their corresponding Lyapunov diagrams computed for two exemplary values of the parameters for which condition~\eqref{eq:warunki2} is satisfied. Namely, we put  $m_1=1, m_2=1$, while $M=6$ and $M=3$ are taken form the first elements of the set~\eqref{eq:warunki2}.  The energy levels were chosen as  $E=E_0+2$, where $E_0$ is the energy minimum. As we can observe,   the plots show highly chaotic behavior of the system precluding its integrability. 
Nevertheless, to prove this fact,  the higher-order variational technique has to be used.

	\section{Coupled pendulums without the gravity \label{sec:3}}
	It is a well-known fact that the classical double pendulum, as well as the coupled pendulums, are integrable in the absence of the gravitational potential~\cite{Levien:93::,Szuminski:20::}. Indeed, there is no restoring torque due to gravity, and therefore the total angular momentum is conserved. On the other hand, it was shown that the zero gravity motions of certain types of multiple pendulum systems are still highly nonlinear and chaotic~\cite{Salnikov:06::,mp:13::c,Szuminski:14::}. This is caused by the presence of constraints and  Hooke's interactions within the systems.  As above, it seems reasonable to study the dynamics and integrability of our model in the absence of gravity.

	\subsection{Canonical transformation}
	For $g=0$ and $a=0$, the system possesses  $\field{S}^1$ symmetry.   Hamiltonian function~\eqref{eq:hamiltonian} depends on the difference of angles only. Therefore, to reduce the number of dependent variables, we perform
	the following canonical transformation
	\begin{equation}
		\label{eq:canonical}
		\begin{aligned}
			&L=\sqrt{\frac{m_1m_2}{m_1+m_2}}\ell,&& P_L=\sqrt{\frac{m_1+m_2}{m_1 m_2}}p_\ell,\\
			&\Theta=\vartheta-\varphi,&&P_\Theta=\frac{m_2p_\vartheta-m_1p_\varphi}{m_1+m_2},\\
			&\Phi=\frac{m_1\vartheta+m_2\varphi}{\sqrt{m_1+m_2}},&& P_\Phi=\frac{p_\vartheta+p_\varphi}{\sqrt{m_1+m_2}}.
		\end{aligned}
	\end{equation}
	Making the above transformation and choosing the new time $t\to \omega_2\tau$, the reduced Hamiltonian now reads
	\begin{equation}
		\begin{split}
			\label{eq:Hamiltonian_reduced}
			H_\text{red.}=\frac{1}{2}\left(\frac{P_L^2}{m}+\frac{P_\Theta^2}{L^2}\right)+\frac{f^2}{2L^2}+L^2\left(\frac{k}{2}+1-\cos\Theta\right).\\
		\end{split}
	\end{equation}
	Here $m$ and $k$ are the new positive  and dimensionless parameters defined as \begin{equation}
		\label{eq:pare}
		m:= \frac{(m_1+m_2)(M+m_1+m_2)}{m_1m_2},\qquad k:=\frac{k_1}{k_2}.
	\end{equation}
	while  $f$  represents the value of the cyclic integral associated to the cyclic coordinate $\Phi$, i.e.,
	\begin{equation}
		\label{eq:alka}
		F=P_\Phi=f.
	\end{equation}
	Thanks to the existence of the linear first integral, the original Hamiltonian~\eqref{eq:hamiltonian} reduces to the system of two degrees of freedom~\eqref{eq:Hamiltonian_reduced} with an additional parameter  $f$. Hamilton's equations of motion take the form
	\begin{equation}
		\begin{cases}
			\label{eq:vh_red}
			\dot L= \dfrac{P_L}{m},\qquad
			\dot P_L=\dfrac{P_\Theta^2}{L^3}+\dfrac{f^2}{L^3}-2L\left(\dfrac{k}{2}+1-\cos\Theta\right),\\[0.3cm]
			\dot \Phi=\dfrac{P_\Theta}{L^2},\qquad \dot P_\Phi=-L^2\sin\Theta.
		\end{cases}
	\end{equation}
Since the phase space of the reduced system is four-dimensional, we provide a quick insight into the system's dynamics using
the Poincar\'e section shown in Fig.~\ref{fig:poinred}. It illustrates that, for the
chosen values of parameters, the system is generally not integrable. The plot displays 
a beautiful coexistence of periodic, quasi-periodic, and chaotic orbits. The
corresponding Lyapunov diagram completes the picture by giving a quantitative
description of chaos. It is evident that the strength of chaos varies among
different chaotic orbits.

	\subsection{Variational equations and~nonintegrability}
	Thanks to the above canonical transformation, we can perform the integrability analysis of the reduced model using various particular solutions. We state the following theorem.
	\begin{theorem}
	In the absence of gravity, and for non-zero parameters $m, k$ and $f$,  the variable-length coupled pendulum system governed by the reduced  Hamiltonian~\eqref{eq:Hamiltonian_reduced} is not integrable in a class of functions meromorphic in coordinates and momenta. 
	\end{theorem}
	\begin{proof}
		System~\eqref{eq:vh_red} possesses the following invariant manifold
		\begin{equation}
			\label{eq:invariant2}
			\scN=\left\{\left(L,P_L,\Theta,P_\Theta\right)\in {\C}^4\,\big{|}\, \Theta=0=P_\Theta\right\}.
		\end{equation}
		Hamiltonian~\eqref{eq:Hamiltonian_reduced} and its corresponding equations of motion~\eqref{eq:vh_red}, restricted to $\scN$, read as follows
		\begin{equation}\scN:
			\begin{cases}
				H_\text{red.}=\dfrac{1}{2}\left(\dfrac{P_L^2}{m}+\dfrac{f^2}{L^2}+k\,L^2\right),\\[0.2cm]
				\dot L=\dfrac{P_L}{m},\quad \dot P_L=\dfrac{f^2}{L^3}-k\, L,\quad \dot \Theta=0,\quad \dot P_\Theta=0.
			\end{cases}
		\end{equation}
		Hence, solving the above equations, we obtain a family of particular solutions $ \vvarphi(t)=(L(t),P_L(t),0,0)$, foliated by a constant energy level $H_\text{red.}=E$. 
		
		Let $[X,P_X,Y,P_Y]^T$ denotes the variations of $[L,P_L,\Theta, P_\Theta]^t$. The variational equations along the particular solution $\vvarphi(t)$, take the form
		\begin{equation}
			\label{eq:var_red}
			\begin{pmatrix}
				\dot X\\
				\dot P_X\\
				\dot Y\\
				\dot P_Y
			\end{pmatrix}=
			\begin{pmatrix}
				0&\frac{1}{m}&0&0\\
				-k-\frac{3f^2}{L^4}&0&0&0\\
				0&0&0&\frac{1}{L^2}\\
				0&0&-L^2&0
			\end{pmatrix}
			\begin{pmatrix}
				X \\ P_X\\ Y\\ P_Y
			\end{pmatrix}
		\end{equation}
		For further analysis, we consider the normal part $(\dot Y,\dot P_Y)$, which can be written as a one second order differential equation~$\mathscr{L}[Y(t)]=0$, where the differential operator $\mathscr{L}$ is defined as follows
		\begin{equation}
			\label{eq:operL}
			\mathscr{L}=D_t^2+2 \left(\frac{\dot L}{L}\right)D_t+1.
		\end{equation}
		It can be easily checked that at zero level of the first integral~\eqref{eq:alka}, the above differential operator $\mathscr{L}$  is solvable in terms of elementary functions. Therefore, for $f=0$ there is no integrability obstacle. Thus, for further analysis, we assume $f\neq 0$.
		\begin{figure*}[t]
			\centering
			\includegraphics[width=0.85	\linewidth]{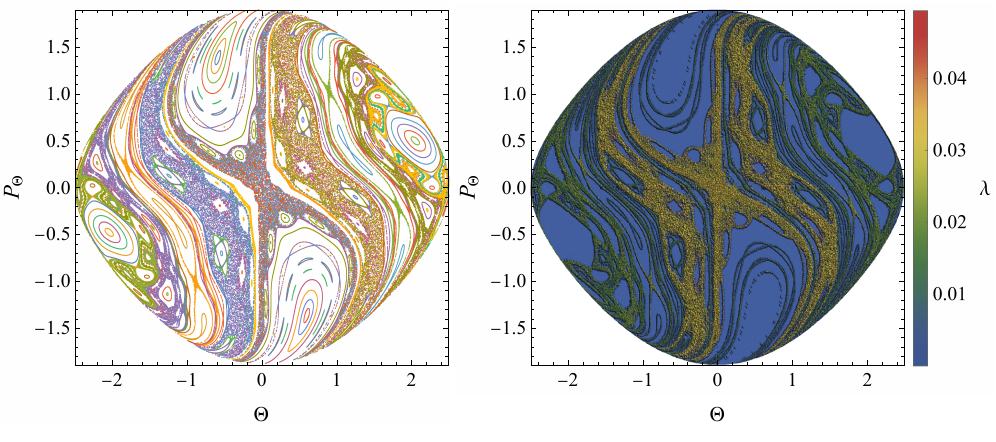}	
			\caption{(Color online)  The Poincar\'e section of the reduced system~\eqref{eq:vh_red} with Hamiltonian~\eqref{eq:Hamiltonian_reduced} and their corresponding Lyapunov diagram. The cross-section plane was specified as $ L=1$ with direction $P_L>0$.   For fixed values of the parameters $m=4, k=1, $ and at the non-zero value of the angular momentum first integral $f=2$, the plot indicates chaotic dynamics. Each color at the Poincar\'e plane corresponds to distinct initial conditions, while in the Lyapunov diagram, the color scale is proportional to the magnitude of the largest Lyapunov exponent $\lambda$. }
			\label{fig:poinred}
		\end{figure*}
	
	Next, we perform the  change of the variable on the equation~\eqref{eq:operL}, namely 
		\begin{equation}
			t\to z=\frac{
				\sqrt{k\,\epsilon}}{f}\, L(t)^2.
		\end{equation}
	The rationalised form of $\mathscr{L}$ is given by
		\begin{equation}
			\label{eq:LZ}
			\mathscr{L}=D_z^2+\frac{1}{2}\left(\frac{2}{z}+\frac{1}{z-1}+\frac{1}{z-\epsilon}\right) D_z-\frac{\omega}{4(z-1)(z-\epsilon)},
		\end{equation}
		where $(\omega,\epsilon)$ are the new parameters defined as follows
		\begin{equation}
			\label{eq:pari}
			\omega:=\frac{m}{k},\qquad \varepsilon:=\frac{2h\left(E+\sqrt{E^2-k\, f^2}\right)}{k\, f^2}-1.
		\end{equation}
		Next, let us apply the classical change of the dependent variable~\eqref{eq:reducedtran}, which  transforms~\eqref{eq:LZ}, into its reduced form
		\begin{eqnarray}
			\label{eq:splitek}
			\mathscr{D}=D_z^2-R(z),
		\end{eqnarray}
		with
		\begin{equation}
			\begin{split}
				\label{eq:rr}
				R(z)&=-\frac{1}{4z^2}-\frac{3}{16(z-1)^2}-\frac{3}{16(z-\epsilon)^2} \\ &+\frac{(5+2\omega)z-2(1+\epsilon)}{8z(z-1)(z-\epsilon)}.
			\end{split}
		\end{equation}
		For $\epsilon \in \R\setminus{\{0,1\}}$ operator $	\mathscr{D}$ has four distinct regular singularities  located at $\Sigma=\{0,1,\epsilon,\infty\}$. The points $\{0,1,\epsilon\}$ are poles of the second order, and the degree of infinity is equal to two.
		
		Following the  Kovacic algorithm, we compute the respective differences of exponents $\Delta_c=\sqrt{1+a_c}$, where $a_c$ are coefficients of the Laurent series expansions of $R(z)$ about $c\in \Sigma$. Their, explicit forms are as follows
		\begin{equation}
			\label{eq:delta}
			\Delta_0=0,\quad \Delta_1=\frac{1}{2},\quad \Delta_\epsilon=\frac{1}{2},\quad \Delta_\infty=\sqrt{\omega+1}.
		\end{equation}
		Now, we can prove the following 
		\begin{lemma}  \label{lem:3}
			Let us assume $\epsilon \in \R\setminus{\{0,1\}}$. Then, for
			\begin{equation}
				\label{eq:condition}
				\sqrt{\omega+1}  \notin \N,
			\end{equation}
			the differential Galois group of operator~\eqref{eq:splitek} with coefficient~\eqref{eq:rr} is $\operatorname{SL(2,\C)}$.
		\end{lemma}
		\begin{proof}
			Taking into account the character of singularities, it appears that all possibilities outlined in Theorem~\ref{lem:m_1_a} must be verified through the Kovacic algorithm. However, according to papers~\cite{Maciejewski:02::,Stachowiak:15::}, if at least one of the differences of exponents is zero, then the differential Galois group cannot be dihedral or finite. 
 Hence, the second and third cases of the algorithm do not hold. The differential Galois group of $	\mathscr{D}$ may either be the triangular group or $\operatorname{SL(2,\C)}$. 
			To check the first possibility, we apply the first case of the Kovacic algorithm. 
			
			Following the algorithm, we compute the auxiliary sets \begin{equation}
				\begin{split}
					&E_0=\left\{\frac{1}{2},\frac{1}{2}\right\},\quad E_1=\left\{\frac{3}{4},\frac{1}{4}\right\},\quad E_\epsilon=\left\{\frac{3}{4},\frac{1}{4}\right\},\\ & E_\infty=\left\{\frac{1}{2}\left(1+\sqrt{\omega+1}\right),\frac{1}{2}\left(1-\sqrt{\omega+1}\right)\right\},
				\end{split}
			\end{equation}
			Next, we check whether there exists families $e_c=(e_0,e_1,r_\epsilon,e_\infty)$  of the Cartesian product $E=E_0\times E_1\times E_\epsilon\times E_\infty$, such that  $d(e)=e_\infty-e_0-e_1-e_\epsilon\in\N_{0}$. We obtain the following distinct possibilities
			\begin{equation}
				d=\frac{1}{2}\left(\sqrt{\omega+1}-p\right)\in \N_{0},\quad \text{for}\quad p=1,2,3.
			\end{equation}
			It is clear that if  condition~\eqref{eq:condition} holds, then  $d\notin \N_{0}$. This ends the proof.
	
		\end{proof} 
		In Lemma~\ref{lem:3} we have  assumed $\epsilon \in \R\setminus{\{0,1\}}$, which indicates that two values of the energy $E$ were excluded. However, if there existed an additional first integral, it would not depend on the energy value, so in particular, it would exist for all generic values of energy. Hence, we were able to safely assume that $\epsilon \neq 0$ and $\epsilon\neq 1$. Nevertheless, there exists a wide range of values of $\omega$ for which  $ \sqrt{\omega+1} =n \in \N$. Hence, we need to analyze the variational equation at these special values of the energy.
		
		We put $\epsilon=1$ and we  assume $\omega=n^{2}-1$. Then, using
		the change of the variable $z\to y=1-z$, we transform~\eqref{eq:LZ} to the Gauss differential operator of the form
		\begin{equation}
			L=D_{y}^{2}+\left(\frac{1}{y}+\frac{1}{y-1}\right)D_{y}-\frac{1-n^{2}}{4y^{2}}.
		\end{equation}
		For the given equation, the respective differences of exponents at singularities $\{0,1,\infty\}$, are as follows		\begin{equation}
			e_{1}=\sqrt{n^{2}-1},\quad e_{2}=0,\quad e_{3}=n.
		\end{equation}
		As $e_{1}$ is irrational for every $n\in \N\setminus\{0,1\}$, and $e_{2}$ is zero, it is evident that neither case of Kimura's theorem (see Appendix) can be satisfied. 
		Hence, from the above analysis, we conclude that the differential Galois group of variational equations~\eqref{eq:var_red} is not Abelian for every  $\omega\in \R^{+}$.  This implies that in the absence of the gravitational potential the coupled pendulum system is not integrable at the non-zero level of the angular momentum cyclic first integral. This concludes the proof.
	\end{proof}
	
	\subsection{Integrability and superintegrability}
			\begin{figure*}
		\centering
			\subfigure[$\,  m = 5,\, k = 1,\, f = 0,\, E=2.25$. Chaotic dynamics]{
			\includegraphics[width=.85 	\linewidth]{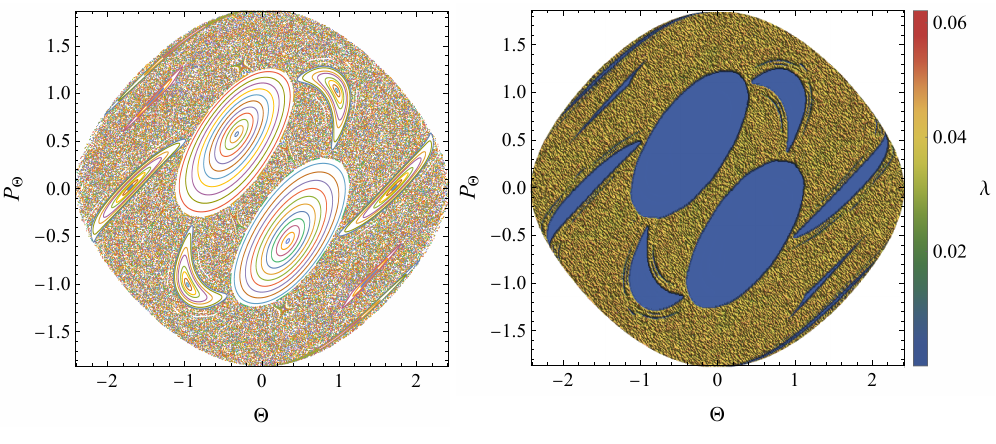}	}
		\subfigure[$\,  m = 4,\, k = 1,\, f = 0,\, E=2.25$. Regular (non-chaotic) dynamics]{
		\includegraphics[width=.85	\linewidth]{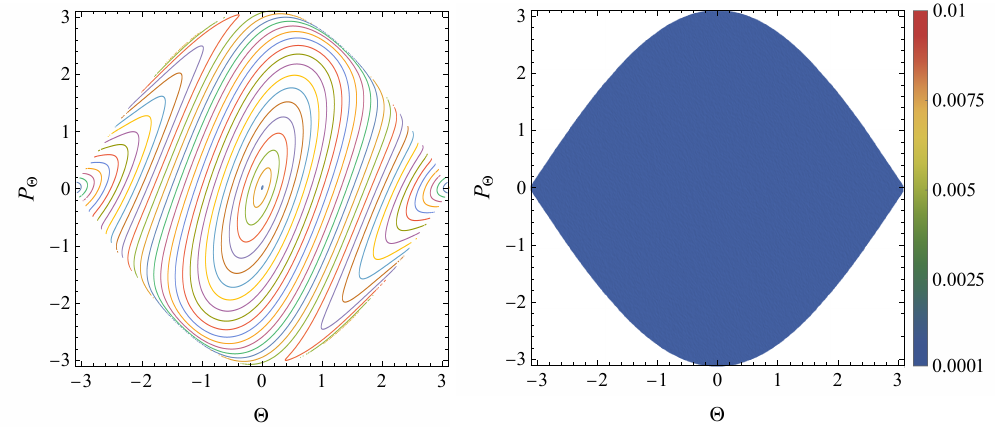}	}
		\caption{(Color online) The Poincar\'e sections of the reduced system~\eqref{eq:vh_red} with Hamiltonian~\eqref{eq:Hamiltonian_reduced} and their corresponding Lyapunov diagram constructed at zero value of the angular momentum first integral $f=0$. The cross-section plane was specified as $L=1$ with direction $P_L>0$.  Each color at the Poincar\'e plane corresponds to distinct initial conditions, while in the Lyapunov diagram, the color scale is proportional to the magnitude of the largest Lyapunov exponent $\lambda$.}
		\label{fig:poinred__int}
	\end{figure*}
In our comprehensive integrability analysis, we have excluded one case, namely,  the zero
value of the cyclic first integral $f=0$.  This exclusion arises from the straightforward solvability of the variational equations. A more detailed analysis (not included)
shows that also the second-order variational equations can be solved in terms of
elementary functions. Hence,  for $f=0$ there are no integrability obstacles.
This observation suggests that the system can indeed be integrable. Since
Hamiltonian~\eqref{eq:Hamiltonian_reduced}  is the system of two degrees of
freedom, it is sufficient to find one additional first integral for its
complete integrability. 
	
To get a quick insight into the dynamics, we generated  the Poincar\'e sections of
the reduced system at zero value of the cyclic first integral $f=0$.
Fig.~\ref{fig:poinred__int} illustrates two exemplary sections. As we can note,
for $m=5, k=1, f=0$, the system is highly chaotic, which precludes its
integrability. Surprisingly, however, by changing the value of the reduced mass
to $m=4$, the dynamic of the system becomes regular. In
Fig.~\ref{fig:poinred__int}(b), we obtain shapely elegant curves with
quasi-periodic orbits. There are no signs of the chaotic nature of the system at
all. Indeed,  the corresponding Lyapunov diagram also indicates the possible presence of
additional first integral since $\lambda\approx 0$, for every initial condition.

The above makes the Lyapunov exponents spectrum a possible indicator for
searching additional first integrals.  Fig.~\ref{fig:indicator} presents
a grid of values of the parameters $(m,k)$, for which we computed Lyapunov
exponents for a  large number of initial conditions uniformly distributed in the available area of the Poincar\'e plane $(\Theta, P_\Theta)$ at
energy levels $E=E_0+2$, where $E_0$ is the energy minimum.
 The color scale of
the dots is proportional to the magnitude of the highest value of the largest Lyapunov exponent at the Poincar\'e plane. Hence, if there is
chaos visible at the Poincar\'e plane, then $\lambda >0 $. Otherwise, if the
Poincar\'e section is regular, without chaotic behavior (as in
Fig.~\ref{fig:poinred__int}(b)), then $\lambda_\text{max}\approx 0$, which
suggest  existence of the first integral. Looking at Fig.~\ref{fig:indicator},
we see that for most values of $(m,k)$ the system is not integrable due to
non-zero values of the largest Lyapunov exponent. However,  for $m=4$    the
situation is quite different. The maximal values of the largest Lyapunov
exponent, $\lambda_\text{max}$, tend to zero for every $k$. The above suggests
the integrability of the system in these cases. 
	
	Indeed,   for $m=4$ and $k$-arbitrary Hamiltonian~\eqref{eq:Hamiltonian_reduced} is  integrable. The additional first integral is a quadratic polynomial with respect to the momenta.
	 The explicit form of the integrable system is as follows 
	\begin{equation}
		\begin{cases}
			\label{eq:calkowalny}
			H=\dfrac{1}{2}\left(\dfrac{P_L^2}{4}+\dfrac{P_\Theta^2}{L^2}\right)+L^2\left(\dfrac{k}{2}+1-\cos\Theta\right),\\[0.2cm]
			F=\left(\dfrac{P_L^2}{4}-\dfrac{P_\Theta^2}{L^2}\right)\cos\Theta-\dfrac{\sin\Theta}{L} P_L\,P_\Theta\\[0.2cm] +L^2\left(k\cos\Theta-4\sin^{2}\dfrac{\Theta}{2}\right).
		\end{cases}
	\end{equation}	 
	\begin{figure}
		\centering
		\includegraphics[width=0.75\linewidth]{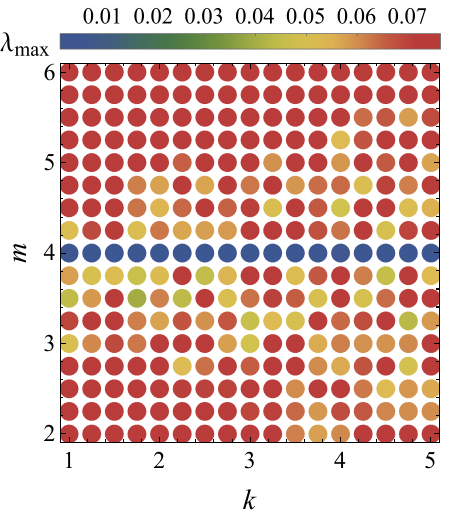}\\
		\caption{(Color online)  Indicators of the first integral. For each value of the parameters $(m,k)$ with $f=0$, Lyapunov diagrams of system~\eqref{eq:vh_red} are computed. The calculations are performed over a $(500\times 500)$ grid of initial conditions $(\Theta_0, P_{\Theta 0})$ with $L_0=1$ at energy levels $E=E_0+2$, where $E_0$ is the energy minimum. The color scale of the marked points at  $(m,k)$-plane corresponds to the highest value of the largest Lyapunov exponent $\lambda$.  The plot indicates that for $m=4$, the system is suspected to be integrable, as  $\lambda_\text{max}\approx 0$. 
		\label{fig:indicator}}
	\end{figure}
		\begin{figure}[htp]
		\centering
			\includegraphics[width=.8	\linewidth]{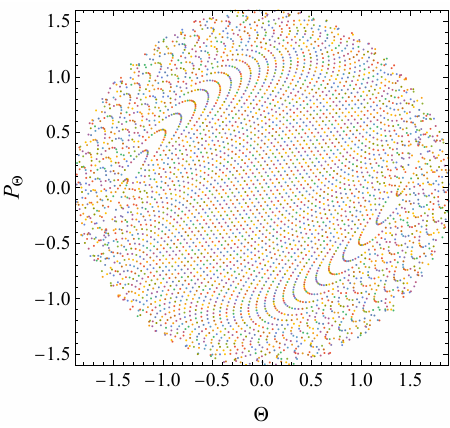}	
		\caption{(Color online) The Poincar\'e section of the reduced system~\eqref{eq:vh_red} constructed for constant values of the parameters: $m=4,\, k=1.5,\, f=0$, at constant energy level $E=2$. The cross-section plane was specified as $L=1$ with direction $P_L>0$.  The plot shows regular dynamics with marked points corresponding to the periodic motion of the system. As there are no quasi-periodic loops (all orbits are closed), the superintgerability of the system is suspected.}
		\label{fig:poinred_super_int}
	\end{figure}
	
			\begin{figure*}
		\centering
		\includegraphics[width=.24 	\linewidth]{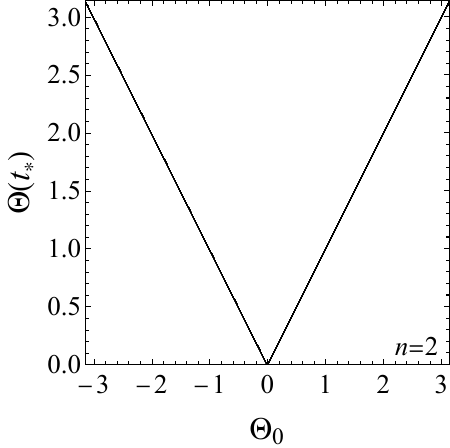}	 \hspace{-.1cm}
		\includegraphics[width=.24 	\linewidth]{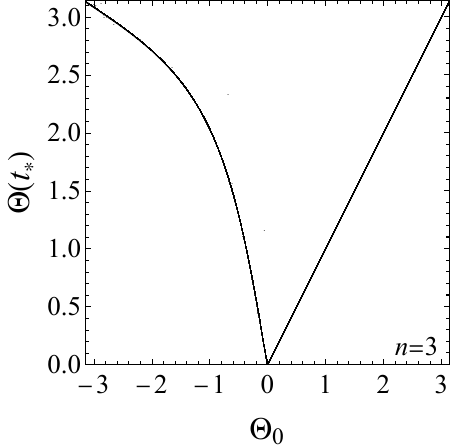}	 \hspace{-.1cm}
		\includegraphics[width=.24 	\linewidth]{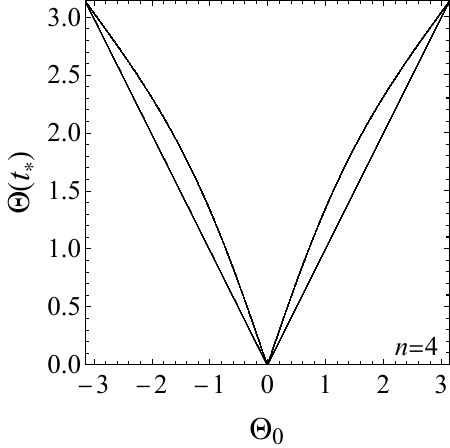}	 \hspace{-.1cm}
		\includegraphics[width=.24 	\linewidth]{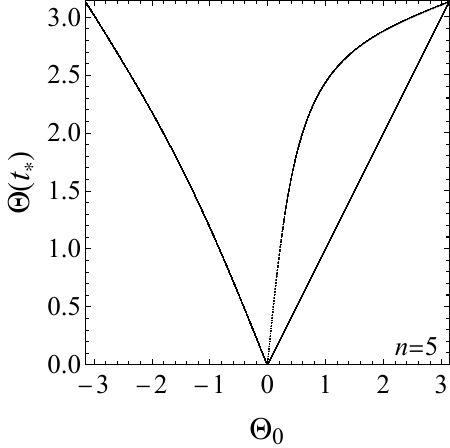}	\\
				\includegraphics[width=.24 	\linewidth]{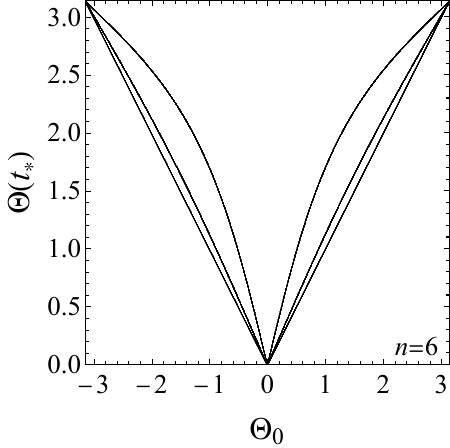}	 \hspace{-.1cm}
		\includegraphics[width=.24 	\linewidth]{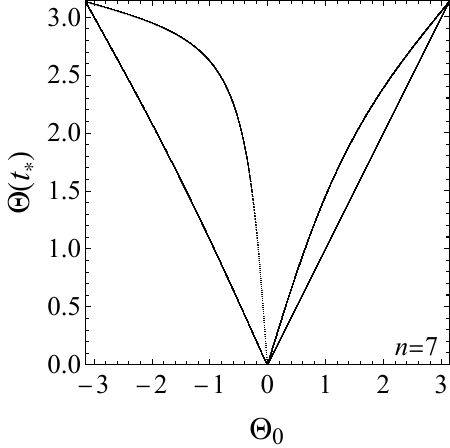}	 \hspace{-.1cm}
		\includegraphics[width=.24 	\linewidth]{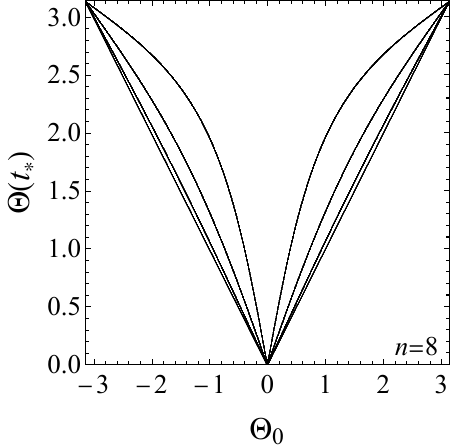}	 \hspace{-.1cm}
		\includegraphics[width=.24 	\linewidth]{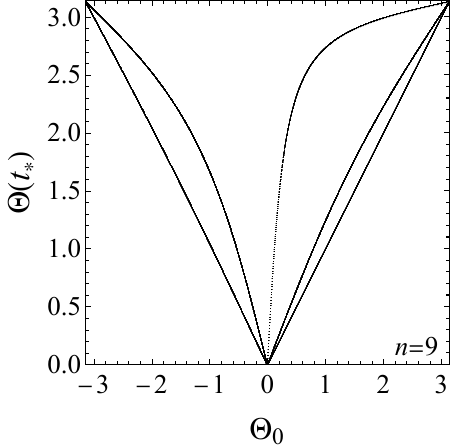}	
		\caption{Phase-parametric diagrams of reduced system~\eqref{eq:vh_red} constructed for $m=4,\,  k=4/(n^2-1)$   at zero value of the cyclic first integral $f =0$,  at variations of the initial angle $\Theta_0$. Here, $\Theta(t_\star)$ denotes the maximal values of $\Theta$ when  $\Theta'(t_\star)= 0$ with  $\Theta''(t_\star)< 0$,   for some~$t_\star$. The diagrams show  only  periodic orbits of the system without any regimes responsible for quasi-periodic motion which indicates its superintegrability.}
		\label{figbifred_super_int}
	\end{figure*}

	Moreover,  for
	\begin{equation}
		\label{eq:parki_super}
k=\frac{4}{n^{2}-1},\quad \text{with}\quad n\in \N\setminus\{0,1\},
	\end{equation}
	system~\eqref{eq:calkowalny} is maximally superintegrable.	The degree of the third first integral increases with the value of $n\in \N\setminus\{0,1\}$, and its general form is as follows
	\begin{equation}
		\label{eq:calk_{general}}
		Q=W_{n} P_{\Theta}+\sum_{i}^{[n/2]}W_{n-2i}\left[A_{i,1}P_{L}+A_{i,2}P_{\Theta}\right]+\delta(n)A_{0},
	\end{equation}
	where 
	\begin{equation}
		W_{n}(L,\Theta)=\frac{n}{L^{n-1}}\left(L\cos\frac{\Theta}{2}P_{L}-2 \sin\frac{\Theta}{2}P_{\Theta}\right)^{n-1}P_{\Theta}.
	\end{equation}
	Hence, $W_{1}=1$ and $W_{0}=0$. 
	Moreover, here $A_{i,1}, A_{i,2}$ and $A_{0}$ are unknown smooth functions of variables $(L,\Theta)$, whereas $[n/2]$ denotes the greatest integer that is less or
	equal to  $n/2$  and $\delta:\R\to\{0,1\}$ is a unit step function
	\begin{equation}
		\delta(n):=\begin{cases}
			1& \text{if}\quad  n\in \text{even},\\
			0& \text{if} \quad n\in \text{odd}.
		\end{cases}
	\end{equation}

	Calculating the Lie derivative of~\eqref{eq:calk_{general}}, we obtain a polynomial in the momenta $(P_{L},P_{\Theta})$. Hence, equating to zero its all coefficients, we get a system of partial differential equations for unknown functions $A_{i,1}, A_{i,2}$ and $A_{0}$.  Solutions of this system give a final form of the seeking first integral.  
 Tab.~\ref{tab:calki}  presents the complementary third first integrals of the superintegrable cases for the first values of $n$. As we can notice, the parity of the complementary first integral depends on $n$. For $n\in \N_\text{even}$ the additional first integral is an even polynomial with respect to the momenta, while for $n\in \N_\text{odd}$ it is an odd function. 
 
Maximally superintegrable systems exhibit the Bertrand property, meaning that all bounded trajectories are closed, and the motion is periodic.  This characteristic is clearly depicted in the  Poincar\'e section and phase-parametric diagrams presented in Figs.~\ref{fig:poinred_super_int}-\ref{figbifred_super_int}. Contrary to the integrable case (see Fig.~\ref{fig:poinred__int}(b)),  superintegrability manifests itself by a finite number of marked points on the Poincar\'e plane. There are no chaotic or even quasi-periodic loops. Each point on this plane corresponds to a distinct initial condition.  Additionally, from the phase-parametric diagrams, we can also deduce that for increasing values of $n$, the frequency ratio of oscillations increases as well.
 
	\begin{table*}[htp]
		\small
		\begin{tabular}{ |l|l|l|l|l|l| } 
			\hline
			$n$ &$k$ & Third first integral& Coefficients\\ \hline
			$2$ & $4/3$& $Q=W_{2}+A_{0}$ & $A_{0}=\frac{8}{3}L^{3}\cos\frac{\Theta	}{2}\sin\Theta$  \\[0.1cm] 
			$3$ & $1/2$ & $Q=W_3+A_{1,1}\, P_{L}+A_{1,2}\,P_{\Theta}$ &$A_{1,1}=2 L^{3}\left(1+\cos\Theta\right)\sin\Theta,\, A_{1,2}=-L^{2}\left(3+\cos\theta-2\cos 2\Theta\right)$ \\[0.1cm]
			$4$ & $4/15$ & $Q=W_{4}+W_{2}\left[A_{1,1}P_{L}+A_{1,2} P_{\Theta}\right]+A_{0}$ & $A_{1,1}=\frac{4}{3}L^{3}\left(1+\cos\Theta\right)\sin\Theta, \, A_{{1,2}}=-\frac{4}{15}L^{2}\left(9+4\cos\Theta-5\cos 2\Theta\right)$\\[0.1cm]
			$5 $ & $1/6$ & $Q=W_{5}+W_{3}\left[A_{1,1}\, P_{R}+A_{1,2}\, P_{\Phi}\right]+A_{2,1}\, P_{R}+A_{2,2}\, P_{\Phi}$& $A_{1,1}=\frac{5}{9}R^{3}\left(2\sin\Theta+\sin 2\Theta\right),\,A_{1,2}=-\frac{10}{9}R^{2}\left(2+\cos\Theta-\cos 2\Theta\right)$\\[0.1cm]
			&& &$A_{2,1}=-\frac{1}{6}R^{5}\sin^{-4}\frac{\Theta}{2}\sin^{5}\Theta, \, A_{2,2}=\frac{4}{9}R^{4}\cos^{4}\frac{\Theta}{2}\left(13-12\cos\Theta\right)$\\[0.1cm]
			$6$ & $4/35$ & $Q=W_{6}+W_{4}\left[A_{1,1}\, P_{R}+A_{1,2}\, P_{\Phi}\right]+W_{2}\left[A_{2,1}\, P_{R}+A_{2,2}\, P_{\Phi}\right]+A_{0}$&$A_{1,1}=R^{3}(1+\cos\Theta)\sin\Theta,\,  A_{1,2}=-\frac{1}{7}R^{2}(15+8\cos\Theta-7\cos 2\Theta)$,\\[0.1cm]
			&&& $A_{2,1}=-\frac{24}{175}R^{5}\sin^{-4}\frac{\Theta}{2}\sin^{5}\Theta, \ A_{2,2}=\frac{768}{1225}R^{4}\cos^{4}\frac{\Theta}{2}(8-7\cos\Theta)$\\[0.1cm]
			&&& $A_{0}=\frac{192}{42875}R^{7}\sin^{-5}\sin^{6}\Theta$\\[0.1cm]
			\hline
		\end{tabular}
		\caption{Table of third first integrals depending on values of $k$.\label{tab:calki}}
	\end{table*}

	\section{Summary and conclusions \label{sec:4}}
One of the
	fundamental problems of the theory of nonlinear dynamical systems and the chaos theory
	is the distinction of integrable models from non-integrable ones. 
	It is a very difficult task because most physical and mechanical systems depend on various parameters, which significantly complicates their integrability analysis. Moreover,  it is a matter of fact that the majority of real-world physical and mechanical systems are not integrable displaying highly chaotic dynamics. Nevertheless, the detection of a  new integrable case of an important physical system is still considered as a significant achievement in mathematics and mechanics.  An illustrative example is Kovalevskaya's highly non-trivial integrable case in rigid body dynamics~\cite{Kovalevskaya:89::}, which was awarded by the Bordin Prize of the French Academy of Sciences.

Recently, various types of pendulums have become the focus of extensive study in nonlinear physics. As mentioned in the introduction, these systems find practical applications in physics and mechanics.
In this paper, we continued our previous work~\cite{Szuminski:20::} by considering the dynamics and integrability of a natural generalization of the coupled pendulum system.  The studied model can be treated as the coupled pendulum system with the variable length, as well as the double-swinging Atwood machine with additional Hooke's iterations. Because the model has three degrees of freedom and depends on parameters, its numerical analysis was quite challenging. For this purpose, we computed Lyapunov exponents diagrams, which gave the quantitative description of chaos. Complementing the Lyapunov diagrams with phase-parametric diagrams allowed us to identify periodic orbits and their count in regions where all Lyapunov exponents approached zero.  Thanks to the existence of the invariant manifold, we were able to construct the Poincar\'e sections, which gave the qualitative description of chaos by showing the beautiful coexistence of periodic, quasi-periodic, and chaotic motion. Moreover, to make the analysis exhaustive, we combined the Poincar\'e sections with the Lyapunov diagrams. As shown this procedure can be effectively used for searching  ,,weak'' chaotic orbits in the Poincar\'e section plane and to measure the strength of chaos. Surprisingly enough, for relatively large energy values, the Poincar\'e sections did not exhibit highly chaotic, almost fully ergodic stages of system dynamics. Instead, we got rich necklace formations of high-order resonance periodic orbits. This observation differs significantly from typical Hamiltonian systems.

	The numerical analysis presented within the paper shows the complex behavior of the system suggesting its nonintegrability. We proved this fact by employing the Morales-Ramis theory and the analysis of the differential Galois group of variational equations. The novelty of our work concerns the fact that we performed the integrability analysis of the Hamiltonian system of three degrees of freedom for which the variational equations transform into the one fourth-order differential equation. To analyze the differential Galois group of this equation, and to prove the nonintegrability of the proposed model, we applied the   Kovacic algorithm of dimension four. We have shown that the variable-length coupled pendulum system is not integrable in the sense of Liouville for almost all values of the parameters. For values~\eqref{eq:warunki} we did not obtain integrability obstructions due to the solvability of variational equations. Initially, we thought that in these cases the system may be integrable -- at least for certain values of the reaming parameters. The numerical analysis, however, suggests the nonintegrability of the system. Therefore, to prove this fact the higher-order variational methods have to be used.   
	
	Finally, in the absence of the gravitational potential, the system has the symmetry $\field{S}^1$,
	and the Hamiltonian depends on the difference of angles only. Therefore, by introducing new variables, we were able to reduce the system to one with two degrees of freedom,  including an additional parameter derived from the momentum first integral. Nevertheless, due to the existence of the constraints and Hooke's interactions between the masses, the system still exhibits complex and chaotic dynamics, which were visualized with the help of the Poincar\'e sections and Lyapunov diagrams. We proved this fact by the analysis of variational equations.
	However, for the zero value of the cyclic first integral, there are no integrability obstacles due to the trivial solvability of variational equations. In this case, we used the Lyapunov exponents diagrams for searching values of the remaining parameters for which the system is suspected to be integrable. 
	 To our knowledge, this was the first attempt to use the  Lyapunov exponents as the indicator of integrable dynamics.  
	 Thanks to that, we find that for $m=4$  the system is integrable, and for certain $k$, the system 
	is even maximally superintegrable with two additional first integrals.  This is an exceptional feature for Hamiltonian systems with more than one degree of freedom. It can be shown that after the appropriate change of variables, this superintegrable system corresponds to the classical two-dimensional anharmonic oscillator. 
	
In summary, in the presented manuscript, we performed a comprehensive analysis of the dynamics and the integrability of the new model of variable-length coupled pendulums.   These results were
obtained with powerful tools, whose applications seem to be of great importance and usefulness to the studies of different pendulums-like systems. Moreover, the considered model reveals different types of dynamics starting from hyperchaos and ending at superintegrability. This makes the variable-length coupled pendulum system as an excellent example of teaching students of Lagrangian and Hamiltonian mechanics and its physical realization can be easily done in the laboratory.  We plan to obtain experimental results concerning its dynamics and compare them with numerical and analytical results obtained in this paper. Therefore, our next work will complete the above theoretical results.
	 \section*{Acknowledgements}
	This research has been  founded by The
	National Science Center of Poland under Grant No.
	2020/39/D/ST1/01632. For the purpose of Open
	Access, the authors have applied a CC-BY public
	copyright license to any Author Accepted Manuscript
	(AAM) version arising from this submission.

	\section{Appendix. Kimura theorem \label{sec:5}}
	The Gauss hypergeometric differential equation, is a homogeneous second-order differential equation with three regular singular points  $z\in\{0,1,\infty\}$, and it is given by	\begin{equation}\label{eq:gausss}
		\dfrac{\mathrm{d}^2\eta}{\mathrm{d}z^2}+\left(\frac{(\alpha+\beta+1)z-\gamma}{z(z-1)}\right)\dfrac{\mathrm{d}\eta}{\mathrm{d}z}+\frac{\alpha\beta}{z(z-1)}\eta=0,
	\end{equation}
	for details see~\cite{Whittaker:35::,Kristensson:12::}. The differences between the exponents 
	\[\rho=1-\gamma,\qquad \sigma=\gamma-\alpha-\beta,\qquad \tau=\beta-\alpha.\] satisfy
	the Fuchs relation
	\[
	\alpha+\alpha'+\gamma+\gamma'+\beta+\beta'=1.
	\]Necessary and sufficient conditions for solvability of the identity
	component of the differential Galois group of the Gauss differential equations~\eqref{eq:gausss} are well-known thanks to the Kimura work~\cite{Kimura:69::}. Let us recall the main theorem. 
	\begin{theorem} [Kimura]
		\label{th:Kimura}
		The identity component of the differential Galois group of the Gauss differential equation~\eqref{eq:gausss} is solvable iff
		\begin{description} 
			\item[A] at least one of the four numbers $\rho+\sigma+\tau$,
			$-\rho+\sigma+\tau$, $\rho+\sigma-\tau$, $\rho-\sigma+\tau$ is an odd
			integer, or
			\item[B] the numbers $\rho$ or $-\rho$  and
			$\sigma$ or $-\sigma$ and $\tau$ or $-\tau$ belong (in an arbitrary order) to some of
			appropriate fifteen families forming the so-called Schwarz's Table \ref{tab:sch_app}.
			\begin{table}[h]
				\begin{tabular}{lllll}
					\toprule
					\text{1}&$1/2+r$&$1/2+s$&$\field{C}$&\\
					\text{2}&$1/2+r$&$1/3+s$&$1/3+p$&\\
					\text{3}&$2/3+r$&$1/3+s$&$1/3+p$&$r+s+p$ even\\
					\text{4}&$1/2+r$&$1/3+s$&$1/4+p$&\\
					\text{5}&$2/3+r$&$1/4+s$&$1/4+p$&$r+s+p$ even\\
					\text{6}&$1/2+r$&$1/3+s$&$1/5+p$&\\
					\text{7}&$2/5+r$&$1/3+s$&$1/3+p$&$r+s+p$ even\\
					\text{8}&$2/3+r$&$1/5+s$&$1/5+p$&$r+s+p$ even\\
					\text{9}&$1/2+r$&$2/5+s$&$1/5+p$&\\
					\text{10}&$3/5+r$&$1/3+s$&$1/5+p$&$r+s+p$ even\\
					\text{11}&$2/5+r$&$2/5+s$&$2/5+p$&$r+s+p$ even\\
					\text{12}&$2/3+r$&$1/3+s$&$1/5+p$&$r+s+p$ even\\
					\text{13}&$4/5+r$&$1/5+s$&$1/5+q$&$r+s+p$ even\\
					\text{14}&$1/2+r$&$2/5+s$&$1/3+p$& \\
					\text{15}&$3/5+r$&$2/5+s$&$1/3+p$&$r+s+p$ even\\
					\hline	\hline
				\end{tabular}
				\caption{Schwarz's table. Here $r,s,p\in\Z$ \label{tab:sch_app}}
			\end{table}
		\end{description}
	\end{theorem}

\end{document}